\documentclass[preprint,12pt]{elsarticle}
\usepackage[latin1]{inputenc}
\usepackage{color}
\usepackage{graphicx}
\usepackage{latexsym}
\usepackage{amsmath}
\usepackage{amssymb}
\usepackage{amsthm}
\usepackage{lineno}
\usepackage{bm}
\usepackage{endfloat}
\usepackage{subfigure}
\usepackage{multirow}
\usepackage{url}
\usepackage{hyperref}

\raggedright 
\newcounter{thm}
\newtheorem{theorem}[thm]{Theorem}
\newtheorem{lemma}[thm]{Lemma}
\newtheorem{prop}[thm]{Proposition}
\newtheorem{coro}[thm]{Corollary}

\theoremstyle{definition}
\newcounter{exa}
\newtheorem{example}[exa]{Example}

\newcommand{\tree}{T} 
\newcommand{\vertex}{V} 
\newcommand{\edge}{E} 
\newcommand{\leaf}{L} 
\newcommand{\str}{\Sigma} 
\newcommand{\erw}{\mathbb{E}} 
\newcommand{\era}{\varepsilon} 
\newcommand{\cov}{\mathrm{Cov}} 
\newcommand{\coa}{\tau} 

\newcommand{\wurzel}{\zeta}
\newcommand{\dens}{\bm{p}} 
\newcommand{\oM}{\overline{M}}
\newcommand{\wM}{\widehat{M}}
\newcommand{\change}[1]{}
\newcommand\steffen[1]{}
\newcommand\volkmar[1]{}
\journal{Math. Biosci.}
\begin{document}
\begin{frontmatter}
\title{An algebraic analysis of the two state Markov model on tripod trees}
\author[otago]{Steffen Klaere\corref{steffen}}
\author[greifswald]{Volkmar Liebscher}
\address[otago]{Department of Mathematics and Statistics, University of Otago, Dunedin, New Zealand}
\address[greifswald]{Institut f\"ur Mathematik und Informatik, Universit\"at Greifswald, Germany}
\cortext[steffen]{sklaere@maths.otago.ac.nz}
\date{\today}


\begin{abstract} Methods of phylogenetic inference use more and more complex models to generate trees from data. However, even simple models and their implications are not fully understood.

Here, we investigate the two-state Markov model on a tripod tree, inferring conditions under which a given set of observations gives rise to such a model. This type of investigation has been undertaken before by several scientists from different fields of research.

In contrast to other work we fully analyse the model, presenting conditions under which one can infer a model from the observation or at least get support for the tree-shaped interdependence of the leaves considered.

We also present all conditions under which the results can be extended from tripod trees to quartet trees, a step necessary to reconstruct at least a topology. Apart from finding conditions under which such an extension works we discuss example cases for which such an extension does not work.
\end{abstract}


\end{frontmatter}

\linenumbers

\section{Introduction}

In phylogeny, one assumes that the relationship of a set of taxonomic units (or taxa) can be visualised by a (binary) tree. The aim is to derive this tree from the observations at the taxa. From a stochastic modelling point of view, one assigns the taxa to the leaves of a (binary) tree, and assumes that the observations (which are usually considered to be i.i.d. over different sites) are the end results of a Markov process along the tree. The goal is to derive the best combination of tree and Markov model to explain the observations.

This work regards the identifiability problem of this inference. It essentially asks whether it is possible that infinite data sets are able to uniquely identify the transitions on the tree and the tree completely. Note that in the present context, identifiability readily leads to consistency of various methods of estimating the parameters of the model \citep[see][Section 2.2 for an overview]{bryant2005b}.

\change{\citet{chang1996} establishes that a certain class of Markov models on a tree can be reconstructed from the  restrictions of its distribution to all triples of leaves. This class of Markov models is characterised by the form of its transition matrices, i.e. that they are invertible, reconstructable from rows and not permutation matrices.} However, usually one only has an estimate of the leaf distribution such a process induces. This leads to the question of whether one can find (simple) conditions to determine whether a taxon distribution comes from a Markov process. In other words, we ask whether we can validate the model, at least if there are infinitely many data points available.

To approach this problem, we consider a very simple model. We assume that our process can take only one of two states for every site, and that the tree is a tripod tree.

Under these restrictions, we can completely describe the map from the taxon distribution to the parameters of the model, including necessary and sufficient conditions on positivity of the parameters. Thereby, no conditions for reversibility of the processes on the edges are needed. The analysis of the model on tripod trees has immediate consequences  for quartet trees. We derive these conditions to exemplify the shortcomings of an extension from tripods to quartets.

Technically, the generic part of this work is already well-known. Initial work on the two state model from psychology can be found in \citet{lazarsfeld1968}. \citet{pearl1986} used these results in artificial intelligence to algorithmically identify the whole tree behind two-state Markov models. Note that identifiability of Markov models especially in phylogeny was studied in \citet{allman2009,allman2008,allman2003,baake1998a,chang1996}. We add to those results the analysis of the degenerate cases, together with a complete analysis of the quartet tree model.

The typical tool (for multi-state models) to identify a subspace of taxon distributions which might come from a Markovian tree model are phylogenetic invariants \citep{allman2008,allman2003,sturmfels2005,lake1987,cavender1987,evans1993}. Those invariants are polynomials in the taxon distribution which are zero for those distributions that are derived from the model of interest. 

\citet{sumner2008} discuss another very interesting set of invariants, the so-called Markov invariants. These are invariants whose value on a tree scales with the determinants of the Markov matrices on the edges. Thus, Markov invariants indicate simple relations between the observations (the distribution of leaf states) and the model (described by the Markov matrices), and provide conditions on the observations based on properties of the model. We will make use of this property in this work.

In the two-state tripod case there is only one, the trivial invariant.  But, not all leaf distributions are derived from the Markov model. In fact, we derive polynomials that vanish on distributions which satisfy the trivial invariant but are not identifiable under the Markov model. To accommodate this observation we suggest incorporating these polynomials into the set of invariants but with the addition that these polynomials do not vanish for identifiable distributions. We discuss degenerate distributions to describe this observation.

Although most of the leaf distributions allow for complex solutions of the model equations, in order for the solution of the algebraic equation to be  parameters of a Markov model additional inequalities must be fulfilled \citep{zwiernik2010,matsen2009,yang2000}. The approach of \citeauthor{matsen2009} is restricted to the Cavender-Farris-Neyman model \citep[CFN][]{cavender1978,farris1973,neyman1971} to accommodate the Hadamard approach \citep{hendy1989,szekely1993}. \citet{yang2000} investigated the CFN model to explore conditions to obtain solutions for different optimisation problems in phylogeny. Extending our approach we recover the inequalities presented in \citet{pearl1986}.

As a final step we investigate how the results for tripod trees extend to trees of four leaves. The results provide a glimpse at what we can expect from the reconstruction from tripods when we have no knowledge of the identifiability of the given taxon distribution.

The structure of this work is as follows:
In Section \ref{sect.model} we describe the general mutation model on a tree, with specialisation to tripod trees coming in Section \ref{sect.tree}. Section \ref{sect.tripod} deals with the complete solution of the two-state tripod tree model. Then, in Section \ref{sect.quartet} we use  these results to analyse the general two-state Markov model on quartet trees. Section \ref{sect:markovinv} discusses the relation between our work and the concept of Markov invariants, and possible extensions of this work. \change{We provide a short summary of our work in Section \ref{sect:discussion}}] For the sake of readability, proofs are presented in \ref{app:proofs}.

\section{\label{sect.model}The Markov model of mutation along a tree}

In this section we introduce the general Markov model and its properties. \citet{pearl1986} nicely motivate this model in the following way. Assume, one is given a set $\leaf$ of taxa and a set of observations from a Markov process $X:\,\leaf\to\{\bm{0},\bm{1}\}$. From these observations one deduces a correlation between the taxa. The assumption is that this correlation can be explained by an underlying (binary) tree $\tree=(\vertex,\edge)$ and an extension $Y:\,\vertex\to\{\bm{0},\bm{1}\}$ of $X$ such that for any pair of taxa there is an interior node such that given the state at the interior node the two taxa are independent. See Fig. \ref{fig:arbitrary} for a depiction of this.

\begin{figure}[htp]
\begin{center}
\includegraphics[width=\textwidth]{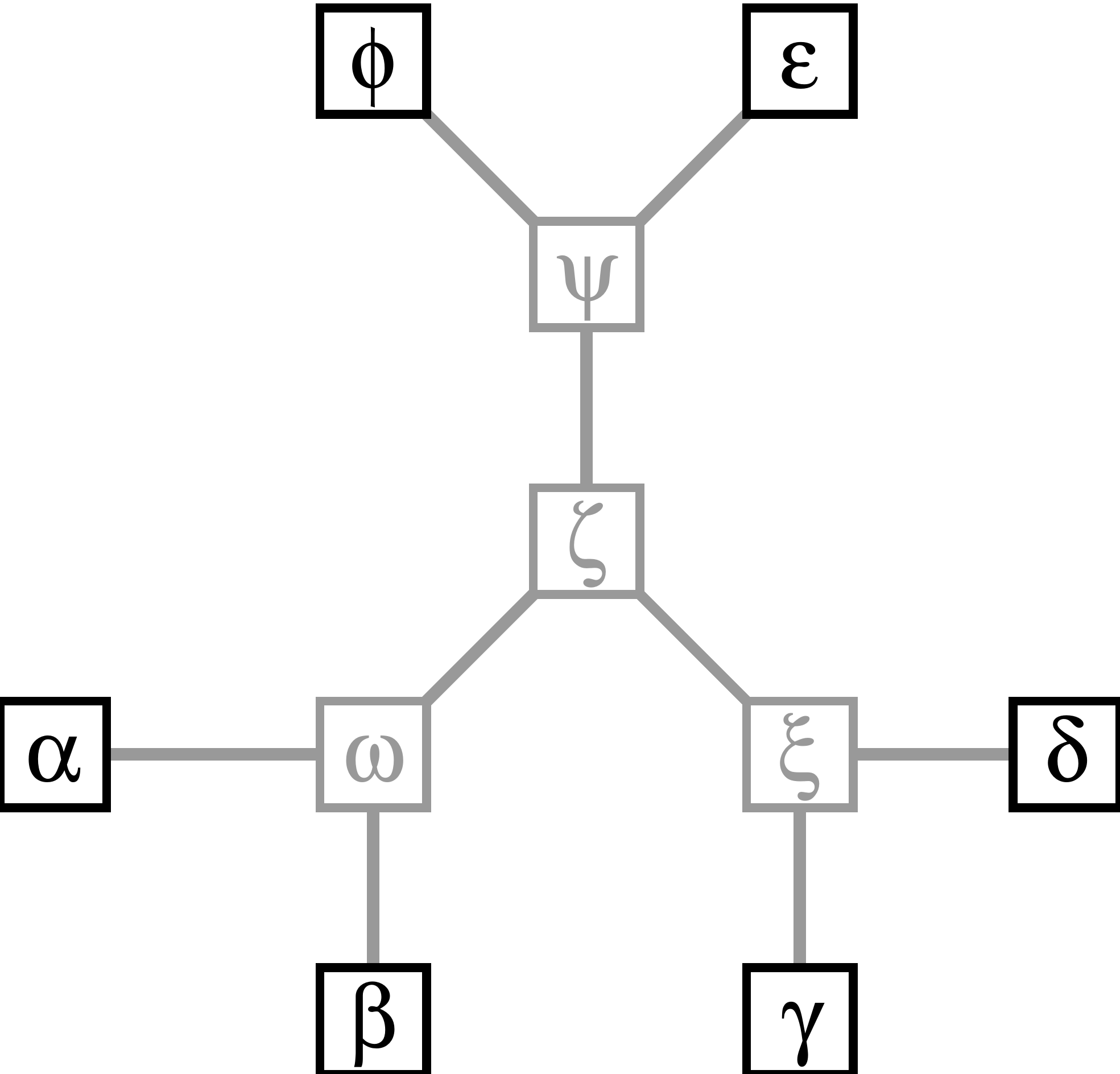}%
\caption{\label{fig:arbitrary}A binary tree with six leaves. Gray lines and nodes describe the hidden part of the process.}
\end{center}
\end{figure}

Let us look closer at the process $Y$. The independence of pairs of taxa given an interior node on the path between them corresponds to the so-called \emph{directed local Markov property} \citep[e.g.,][Chapter 2]{lauritzen1996}. For this property one has to identify a node $\wurzel\in\vertex$ as the root of the tree and direct all edges away from $\wurzel$. Thus, our tree becomes a directed acyclic graph, and for every node $\beta\in\vertex\setminus\{\wurzel\}$ there is a parent node $\alpha\in\vertex$ (with respect to the root), such that $(\alpha,\beta)\in\edge$. Further, for each node $\beta\in\vertex$ one defines the set its descendants as those nodes $\alpha$ for which the path from the root to $\alpha$ passes through $\beta$. The non-descendants are then the nodes that are neither descendants nor parents.

The directed local Markov property states that conditioned on the state of its parent node the state of a node $\alpha\in\vertex$ is independent of the states of its non-descendants. With this property the joint distribution $\widetilde{\dens}^Y$ has the {\it factorisation property}, i.e. for the joint state $\bm{\chi}\in\{\bm{0},\bm{1}\}^{|\vertex|}$ we get
\begin{linenomath}
\begin{equation}\label{eq:factorization}
\widetilde{p}^Y_{\bm{\chi}}=\Pr[Y_\wurzel=\chi_\wurzel]\prod_{(\alpha,\beta)\in\edge}
\Pr[Y_\beta=\chi_\beta|Y_\alpha=\chi_\alpha]=
q^\wurzel_{\chi_\wurzel}\prod_{(\alpha,\beta)\in\edge}M^{\alpha\beta}_{\chi_\alpha\chi_\beta}.
\end{equation}
\end{linenomath}
Here, the marginal distribution $\bm{q}^\wurzel$ corresponds to the initialisation of the process, i.e. $q^\wurzel_z$ is the probability that the process attains state $z\in\{\bm{0},\bm{1}\}$ at the root. The transition matrices $(\bm{M}^e)_{e\in\edge}$ describe the way the process progresses along an edge.  E.g., for an edge $(\alpha,\beta)\in\edge$ the term $M^{\alpha\beta}_{ab}$ is the probability that the character  $a$ at node $\alpha$ is mutated into character $b$ at node $\beta$.

In summary, the joint probability distribution $\widetilde{\dens}^Y$ is given by the marginal distribution $\bm{q}^\wurzel$ and the transition matrices $(\bm{M}^e)_{e\in\edge}$, and thus such a Markov process is completely characterised by these parameters. We will call $\bm{q}^\wurzel$ and $(\bm{M}^e)_{e\in\edge}$ the \textit{process parameters}.

In general, the actual position of the root node $\wurzel$ is not important for \eqref{eq:factorization}, i.e. $\wurzel$ can be chosen arbitrarily from $\vertex$, including a leaf \citep[e.g.,][]{allman2003}.

We only have partial knowledge on the realisations of the process $Y$ through the process $X$ on the leaves. The joint distribution $\dens^X$ of $X$ can then be inferred from \eqref{eq:factorization} using the law of total probability. Let $\bm{x}\in\{0,1\}^{|\leaf|}$ denote the joint state at the leaves. Then
\begin{linenomath}
\begin{equation}\label{eq:leaftotal}
p^X_{\bm{x}}=\sum_{\substack{\bm{\chi}\in\vertex\\\bm{\chi}|_\leaf=\bm{x}}}\widetilde{p}^Y_{\bm{\chi}}=\sum_{\substack{\bm{\chi}\in\vertex\\\bm{\chi}|_\leaf=\bm{x}}}q^\wurzel_{\chi_\wurzel}\prod_{(\alpha,\beta)\in\edge}M^{\alpha\beta}_{\chi_\alpha\chi_\beta}.
\end{equation}
\end{linenomath}
Note that under the assumption that $X$ comes from a reversible Markov process $Y$ \citet{chang1996} proved that all process parameters can be recovered from all the distributions of the restrictions of $X$ to arbitrary triples of taxa.

If we find process parameters for a joint taxon distribution $\dens$ then we call $\dens$ \emph{tree decomposable}. If the obtained process parameters are unique (up to model-specific symmetries), we call $\dens$ \emph{algebraically identifiable}, and if further the process parameters are marginal and transition probabilities, then $\dens$ is called \emph{stochastically identifiable}. Clearly, any stochastically identifiable distribution is also algebraically identifiable.

Looking at \eqref{eq:leaftotal} we realise that verifying the tree decomposability of a distribution $\dens$ is equivalent to solving a polynomial equation system of $2^{|L|}-1$ independent equations in $4|L|-5$ variables. We observe that the Markov equations are overdetermined for $|L|>3$, i.e. the space of tree decomposable distributions is a proper subspace of the space of all distributions. From this we conclude, that there are conditions that define a tree decomposable distribution. These conditions are generally known as \emph{invariants}, polynomials in $2^{|L|}-1$ variables whose roots are distributions that are tree decomposable. One example of an invariant is
\begin{linenomath}
\begin{equation}\label{eq:trivialinv}
\sum_{\bm{x}\in\{0,1\}^{|\leaf|}}p_{\bm{x}}=1,
\end{equation}
\end{linenomath}
i.e. all probabilities sum to one. This is fittingly called the \emph{trivial invariant}.  \citet{allman2008} provide a complete set of invariants for trees of arbitrary size under a two-state-model, and observe that for complete identification the knowledge of the restrictions to six taxa are necessary.

However, as pointed out in multiple publications \citep[e.g.,][]{pearl1986,matsen2009} such invariants are not sufficient to guarantee tree identifiability. In particular, additional inequalities are needed.

Here, we are not only interested in recapturing invariants and inequalities. In addition, we also investigate those distributions that are not algebraically identifiable or not tree decomposable at all to discuss their impact on invariant-based inference.

\section{\label{sect.tree}General properties of a Markov model on a tripod tree}

The starting point of our analysis is the tripod tree $\tree$ with taxa $\alpha,\beta,\gamma$, interior node $\wurzel$ and edges $(\wurzel,\alpha),\,(\wurzel,\beta),\,(\wurzel,\gamma)$ (see Fig. \ref{fig:tripod}). This is the only labeled topology for three taxa. Hence any inference will be process- and not topology-related. \citet{allman2003} select a taxon as the root for their approach. We will place the root at the interior node for the symmetry this provides in the tree equations. 

\begin{figure}[htp]
\includegraphics[width=\textwidth]{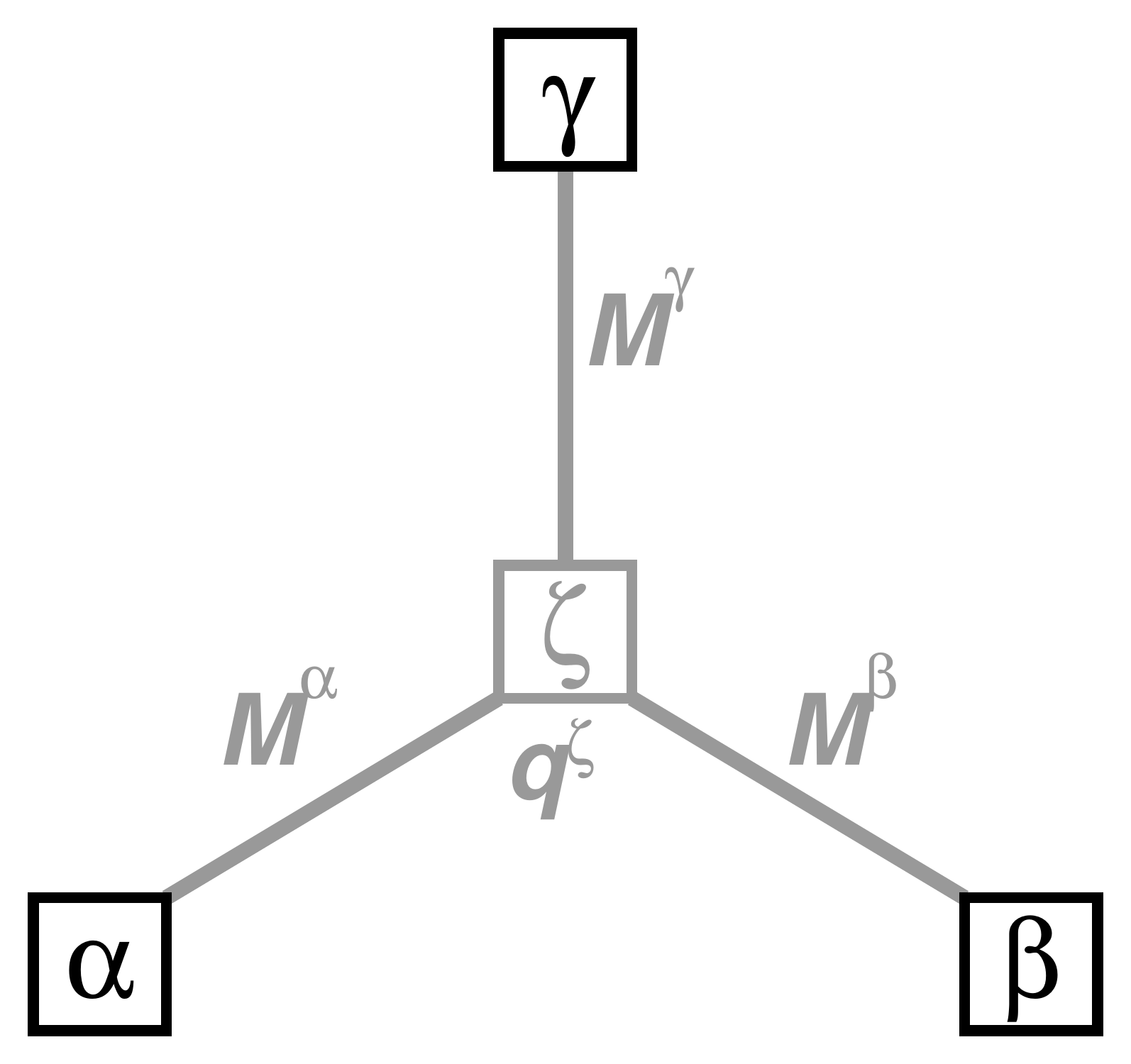}
\caption{\label{fig:tripod}The tripod tree $\tree$.}
\end{figure}

As stated in the previous section, if the joint distribution $\dens$ of $X_\alpha,X_\beta,X_\gamma$ comes from a Markov process then there are parameters $\bm{q}^\wurzel,\,\bm{M}^\alpha,\,\bm{M}^\beta,\,\bm{M}^\gamma$ such that the Markov equations \eqref{eq:leaftotal} are satisfied. On a tripod tree these equations are the \emph{tripod equations}
\begin{linenomath}
\begin{equation}\label{eq:tripodtotal}
p_{abc}=q^\wurzel_1M^\alpha_{1a}M^\beta_{1b}M^\gamma_{1c}+(1-q^\wurzel_1)M^\alpha_{0a}M^\beta_{0b}M^\gamma_{0c},\quad a,b,c\in\{0,1\}.
\end{equation}
\end{linenomath}
As before we call $\dens$ \emph{tree decomposable}, if there are parameters, \emph{algebraically identifiable}, when the parameters are unique (up to some symmetries discussed later), and \emph{stochastically identifiable} if the parameters are unique and proper marginal and transition probabilities.

The works of \citet{lazarsfeld1968} and \citet{pearl1986} were mainly interested in inferring conditions under which a triplet distribution is stochastically identifiable. While recovering their results we also investigate tree decomposability and algebraic identifiability in order to describe their impact on invariant-based inference.

For three taxa the only invariant is the trivial invariant. Thus, one could expect that all triplet distributions are tree decomposable. As we will see later, this is not the case. In fact, we will present polynomials whose roots satisfy the trivial invariant but are not tree decomposable.

\subsection{Statistics for binary models}\label{sect:statbin}

Following \citet{pearl1986}  we identify the symbols $0$ and $1$ with their actual integer counterparts. This permits us to introduce a set of terms that are very helpful for later steps of the analysis. We start by introducing the following abbreviations:
\begin{linenomath}
\begin{align*}
\era_{\alpha\beta\gamma}&:=\erw X_\alpha X_\beta X_\gamma =\Pr[X_\alpha=1,X_\beta=1,X_\gamma=1]=p_{111},\\
\era_{\alpha\beta}&:=\erw X_\alpha X_\beta=\Pr[X_\alpha=1,X_\beta=1]=p_{11\str}=p_{110}+p_{111},\\
\era_\alpha&:=\erw X_\alpha=\Pr[X_\alpha=1]=p_{1\str\str}=p_{100}+p_{101}+p_{110}+p_{111}.
\end{align*}
\end{linenomath}
The symbols $p_{11\str}$ and its modifications $p_{1\str1}$ etc. are direct consequences of the application of the law of total probability to the equation system \eqref{eq:tripodtotal}. These terms are also known as \emph{marginalisations} leading to a removal of a random variable from consideration by summing over its states. This linear modification means we can study the tripod equations \eqref{eq:tripodtotal} also in terms of its marginalisations.

In the case of the binary model the above symbols $\era_A$ for all $A\in\leaf$ correspond to the joint mean of the random variables for the taxa in $A$. Using these definitions we can introduce simple terms which correspond to the covariances between the set of random variables:
\begin{linenomath}
\begin{align*}
\coa_{\alpha\beta}&:=\cov[X_\alpha,X_\beta]=\erw X_\alpha X_\beta-\erw X_\alpha\erw X_\beta,
\end{align*}
\end{linenomath}
with equivalent definitions for $\coa_{\alpha\gamma}$ and $\coa_{\beta\gamma}$. Of further interest are the following terms ($c\in\{0,1\}$)
\begin{linenomath}
\begin{align*}
\coa_{\alpha\beta|c}&:=p_{11c}p_{\str\str c}-p_{1\str c}p_{\str1c},
\end{align*}
\end{linenomath}
with equivalent definitions for $\coa_{\alpha\gamma|b},\,b\in\{0,1\}$ and $\coa_{\beta\gamma|a},\,a\in\{0,1\}$. These terms are actually multiples of the conditional covariances, $\mathrm{Cov}[X_\alpha,X_\beta|X_\gamma=c]=\coa_{\alpha\beta|c}/p_{\str\str c}$. \change{We will use this linear relationship to talk about conditional covariance when we are regarding terms like $\coa_{\alpha\beta|c}$.} Finally, we also introduce the three-way covariances 
\begin{linenomath}
\begin{align*}
\coa_{\alpha\beta\gamma}&:=\cov[X_\alpha,X_\beta,X_\gamma]=\erw(X_\alpha-\erw X_\alpha)(X_\beta-\erw X_\beta)(X_\gamma-\erw X_\gamma)\\
&=\era_{\alpha\beta\gamma}-\era_\alpha\era_{\beta\gamma}-\era_\beta\era_{\alpha\gamma}-\era_\gamma\era_{\alpha\beta}+2\era_\alpha\era_\beta\era_\gamma.
\end{align*}
\end{linenomath}
For a review on covariance for more than two random variables see e.g. \citet{rayner2009}. The term $\coa_{\alpha\beta\gamma}$ describes the interactions of the three leaves considered. \citet{sumner2008} call this term a \emph{stangle}, a stochastic tangle, highlighting its relation to entangled states of qbits in quantum mechanics. The three-way covariances are zero in the case of symmetric models like CFN, which also reflects the findings in \citet{baake1998a}. However, for more complex models the three-way covariances are needed as indicated by the findings of \citet{chang1996}.

Since covariances are a measure of interdependence of random variables, and because the identification of a tree and a Markov model is an interpretation of the interdependence in terms of hidden variables and conditional independence, looking at these covariances is a very logical way to verify whether or not such an interpretation is admissible. Using these terms we can immediately propose a useful property.
\begin{lemma}\label{lemma:flip}
Let $\dens$ denote the joint probability for binary random variables $X_\alpha,\,X_\beta$ and $X_\gamma$.  If we flip the state in one taxon, then we flip the signs in its pairwise covariances. E.g., if $X_\alpha\mapsto 1-X_\alpha$, then  $\coa_{\alpha\beta}\mapsto -\coa_{\alpha\beta}$, $\coa_{\alpha\gamma}\mapsto -\coa_{\alpha\gamma}$ $\coa_{\beta\gamma}\mapsto \coa_{\beta\gamma}$.
\end{lemma}

One immediate consequence of this observation is that the product $\coa_{\alpha\beta}\coa_{\alpha\gamma}\coa_{\beta\gamma}$ always has the same sign no matter how much we flip states.

\subsection{Tree properties}

In this section we assume that $\dens$ is tree decomposable and regard some immediate consequences. We will later see that these conditions are necessary for identifiability but not sufficient. Nevertheless, these conditions provide some immediate insights for it.

\begin{lemma}\label{lem:indi}
\begin{enumerate}
\item If a triplet distribution $\dens$ is tree decomposable on $\tree$ with $\coa_{\alpha\beta}=0$, then also $\coa_{\alpha\beta\gamma}=0$ and $\coa_{\alpha\gamma}=0$ or $\coa_{\beta\gamma}=0$. 
\item If a triplet distribution $\dens$ is stochastically identifiable then the product $\coa_{\alpha\beta}\coa_{\alpha\gamma}\coa_{\beta\gamma}$ is non-negative.
\end{enumerate}
\end{lemma}

The non-negativity of the product has already been verified by \citet{lazarsfeld1968}. With Lemma \ref{lemma:flip} it is not complicated to derive  that on a star tree (with arbitrary number of leaves) there always is a state flipping such that all pairs of leaves are positively correlated.

\begin{coro}\label{coro:positiveflipping}
Suppose we are given a  stochastically identifiable distribution $\dens$ on a  tree with finite  leaf set $L$ such that  the pairwise covariances do not vanish, i.e., $\coa_{\alpha\beta}\ne0$ for all $\alpha,\beta\in L$. Then there exists a set of leaves $L_0\subset L$ such that flipping the states of the leaves in $L_0$  yields all covariances  $\coa_{\alpha\beta}$, $\alpha,\beta\in L$, being positive.
\end{coro}

Lemma \ref{lem:indi}(1) occurs exactly if $X_\alpha$ or $X_\beta$ is independent of the remaining random variables. It also implies the following:

\begin{coro}\label{coro:notreepair}
A triplet distribution $\dens$ with $\coa_{\alpha\beta}=0$ but $\coa_{\alpha\gamma}\ne0$ and $\coa_{\beta\gamma}\ne0$ is not tree decomposable.
\end{coro}

Thus we already see, that the trivial invariant does not characterise tree decomposable distributions in this setting. The following example shows that such cases can be easily constructed.

\begin{example}\label{ex:notreepair}
Triplet distributions of type
\begin{linenomath}
\begin{align*}
\dens&=(p_{000},p_{001},p_{010},p_{011},p_{100},p_{101},p_{110},p_{111})\\
&=(4-x,x,2,2,2,2,2,2)/16,\quad x\in[0,4]\setminus\{2\},
\end{align*}
\end{linenomath}
yield $\coa_{\alpha\beta}=0$ but $\coa_{\alpha\gamma}=\coa_{\beta\gamma}=(2-x)/32$ and hence are not tree decomposable. In fact, for  binary variables  a much more  complicated graphical model with more ``inner'' nodes and edges is needed  to explain  theses covariances.
\end{example}

\section{\label{sect.tripod}Solving the tripod equations}

In this section we are given a triplet distribution $\dens$ and infer conditions under which it is algebraically identifiable. For each case we will present an example. 

\subsection{The algebraic solution}

As has been pointed out multiple times, the only invariant in the tripod case is the trivial invariant. In other words, the ``set'' of invariants for a tripod tree is satisfied by all triplet distributions. However, as we have seen in Corollary \ref{coro:notreepair} there are triplet distributions that are not tree decomposable even though they satisfy the trivial invariant. Thus executing the actual decomposition, i.e. finding a solution for the tripod equations not only provides complete forms for the parameters but is also helpful to identify further cases. The first task is to clarify up to which level of uniqueness the decomposition of a triplet distribution can be attained. To do this we look at the implications of a state-flip at the root.

\begin{lemma}\label{lem:stateflip}
If a triplet distribution $\dens$ is tree decomposable with parameters $\bm{q}^\wurzel,\,\bm{M}^\alpha,\,\bm{M}^\beta,\,\bm{M}^\gamma$ then it is also tree decomposable for parameters $\bm{\widetilde q}^\wurzel,\,\bm{\widetilde M}^\alpha,\,\bm{\widetilde M}^\beta,\,\bm{\widetilde M}^\gamma$ with ${\widetilde q}^\wurzel_z=q^\wurzel_{1-z},\,{\widetilde M}^\alpha_{za}=M^\alpha_{(1-z)a},\,{\widetilde M}^\beta_{zb}=M^\alpha_{(1-z)b},\,{\widetilde M}^\gamma_{zc}=M^\gamma_{(1-z)c}$.
\end{lemma}

Hence, except for the case where everything is equal to $1/2$, there will always be at least two sets of parameters that decompose a triplet distribution $\dens$. In terms of molecular evolution one can view these solutions as having either few mutations ($M^\delta_{z(1-z)}<M^\delta_{zz},\,\delta$ leaf) or many mutations ($M^\delta_{z(1-z)}>M^\delta_{zz},\,\delta$ leaf) for the other. \citet{chang1996} addressed the problem of symmetric solutions by introducing matrix categories that are \emph{reconstructible from rows}. One such class consists of diagonally dominant matrices, i.e. $M^\delta_{zz}>M^\delta_{z(1-z)}$ for all leaves and $z\in\{0,1\}$. If only these two sets of parameters exist then we will regard the associated distribution as algebraically identifiable. It should be noted that the set of symmetric solutions increases with the number of parameters, i.e. each possible permutation of the states at the root yields a new solution. This fact has also been observed by \citet{chang1996} in the case of the time-continuous Markov model.

Next, we present conditions under which $\dens$ is algebraically identifiable and present the closed form for the parameters.

\begin{theorem}\label{thm:alg}
Let $\dens$ denote a triplet distribution and assume
\begin{linenomath}
\begin{equation}\label{eq:alg}
\coa_{\alpha\beta}\coa_{\alpha\gamma}\coa_{\beta\gamma}\ne0,\quad
\coa_{\alpha\beta}\coa_{\alpha\gamma}\coa_{\beta\gamma}\ne-
\bigg(\frac{\coa_{\alpha\beta\gamma}}{2}\bigg)^2.
\end{equation}
\end{linenomath}
Then $\dens$ is algebraically identifiable. The associated parameters have the following form:
\begin{linenomath}
\begin{equation}\label{eq:sol}
\begin{split}
q^\wurzel_1&=\frac{1}{2}-\frac{\coa_{\alpha\beta\gamma}}{2\sqrt{\chi}},\\
M^\alpha_{01}&=\era_\alpha+\frac{\coa_{\alpha\beta\gamma}-\sqrt{\chi}}{2\coa_{\beta\gamma}},\quad
M^\beta_{01}=\era_\beta+\frac{\coa_{\alpha\beta\gamma}-\sqrt{\chi}}{2\coa_{\alpha\gamma}},\quad
M^\gamma_{01}=\era_\gamma+\frac{\coa_{\alpha\beta\gamma}-\sqrt{\chi}}{2\coa_{\alpha\beta}},\\
M^\alpha_{11}&=\era_\alpha+\frac{\coa_{\alpha\beta\gamma}+\sqrt{\chi}}{2\coa_{\beta\gamma}},\quad
M^\beta_{11}=\era_\beta+\frac{\coa_{\alpha\beta\gamma}+\sqrt{\chi}}{2\coa_{\alpha\gamma}},\quad
M^\gamma_{11}=\era_\gamma+\frac{\coa_{\alpha\beta\gamma}+\sqrt{\chi}}{2\coa_{\alpha\beta}},
\end{split}
\end{equation}
\end{linenomath}
where $\chi=\coa_{\alpha\beta\gamma}^2+4\coa_{\alpha\beta}\coa_{\alpha\gamma}\coa_{\beta\gamma}$.
\end{theorem}

Note that \citet{pearl1986} presented a similar solution for the parameters. Looking at the parameters in \eqref{eq:sol} we see that algebraically the conditions in \eqref{eq:alg} prevent division by zero. Together with the trivial invariant we can thus claim that the space of algebraically identifiable triplet distributions is given by $\mathcal{S}\setminus(\mathcal{S}_0\cup\mathcal{S}_1)$ with
\begin{linenomath}
\begin{align*}
\mathcal{S}&:=\{\dens\in\mathbb{R}^8_+:\,p_{000}+\dots+p_{111}=1\},\\
\mathcal{S}_0&:=\{\dens\in\mathcal{S}:\,\coa_{\alpha\beta}\coa_{\alpha\gamma}\coa_{\beta\gamma}=0\},\\
\mathcal{S}_1&:=\{\dens\in\mathcal{S}:\,\coa_{\alpha\beta\gamma}^2+4\coa_{\alpha\beta}\coa_{\alpha\gamma}\coa_{\beta\gamma}=0\}.
\end{align*}
\end{linenomath}

Considering \eqref{eq:alg} and Lemma \ref{lem:indi}(2) we see that triplet distributions with $\coa_{\alpha\beta}\coa_{\alpha\gamma}\coa_{\beta\gamma}<0$ are only algebraically, but not stochastically identifiable. In fact, for $-\coa_{\alpha\beta\gamma}^2<4\coa_{\alpha\beta}\coa_{\alpha\gamma}\coa_{\beta\gamma}<0$ we get real-valued parameters \change{that are not all probabilities}, and for $4\coa_{\alpha\beta}\coa_{\alpha\gamma}\coa_{\beta\gamma}<-\coa_{\alpha\beta\gamma}^2$ we get a set of complex-valued parameters.

The following example presents such distributions.

\begin{example}\label{ex:negative}
Regard the distributions
\begin{linenomath}
\begin{align*}
\dens_1&=(6,7,2,1,1,1,4,5)/27,\quad
\dens_2=(6,7,1,2,1,1,4,5)/27,\quad
\dens_3=(6,6,2,2,1,1,4,5)/27.
\end{align*}
\end{linenomath}
All three distributions satisfy the conditions \eqref{eq:alg}, i.e. they are algebraically identifiable. For $\dens_1$ the covariance $\coa_{\beta\gamma}$ is negative and the other two positive, while for $\dens_2$ we have $\coa_{\alpha\gamma}$ negative and the other two positive. The distribution $\dens_3$ has only positive pairwise covariances. 

The parameters for $\dens_1$ are real-valued, the parameters for $\dens_2$ are complex-valued and $\dens_3$ is stochastically identifiable.

Though this example is artificial it indicates just how sensitive the model is to misreads in alignments. E.g., the difference between $\dens_1$ and $\dens_3$ could be seen as reading the pattern $011$ under $\dens_3$ as pattern $001$ under $\dens_1$.
\end{example}

\subsection{Stochastically identifiable distributions}

The next step is to determine conditions under which a distribution satisfying \eqref{eq:alg} is stochastically identifiable. These conditions should correspond to the conditions given by \citet[Theorem 1]{pearl1986}.

Example \ref{ex:negative} dealt with $\coa_{\alpha\beta}\coa_{\alpha\gamma}\coa_{\beta\gamma}<0$. However, as the following example shows, positivity of the product does not necessarily yield stochastic identifiability.

\begin{example}\label{exa:pos}
The tripod distribution
\begin{linenomath}
\begin{equation*}
\dens=(68,0,20,12,20,12,17,51)/200
\end{equation*}
\end{linenomath}
yields positive covariances for all three pairs but also $M^\gamma_{01}=-1/20$, i.e. not a probability.
\end{example}

The example contains a pattern of expected zero occurrence. From the tripod equations we conclude that a stochastically identifiable distribution is strictly positive, thus this example is slightly contrived. However, as Example \ref{ex:notreepair} showed, a strictly positive triplet distribution is not necessarily stochastically identifiable either.

In order to get necessary and sufficient conditions on a triplet distribution to be stochastically identifiable we need to go back to the parameters in \eqref{eq:sol} and bound them accordingly. This yields:

\begin{theorem}\label{thm:stoch}
A triplet distribution $\dens$ is stochastically identifiable if and only if after suitable state flips the following inequalities hold
\begin{linenomath}
\begin{equation}\label{eq:probs}
\begin{split}
\coa_{\alpha\beta}>0,&\quad\coa_{\alpha\beta|0}\ge0,\quad\coa_{\alpha\beta|1}\ge0,\\
\coa_{\alpha\gamma}>0,&\quad\coa_{\alpha\gamma|0}\ge0,\quad\coa_{\alpha\gamma|1}\ge0,\\
\coa_{\beta\gamma}>0,&\quad\coa_{\beta\gamma|0}\ge0,\quad\coa_{\beta\gamma|1}\ge0.
\end{split}
\end{equation}
\end{linenomath}
\end{theorem}

In other words, the direction of the correlation between a pair of leaves shall not be influenced by the third leaf. \change{Note that \citet{pearl1986}, Theorem 1 presents analogous conditions on $\dens$.} With this we can summarise that a triplet distribution is stochastically identifiable if it is in $\mathcal{S}\setminus(\mathcal{S}_0\cup\mathcal{S}_1)$ and there is a state flip such that \eqref{eq:probs} is satisfied.

\begin{example}
The tripod distribution $\dens$ from Example \ref{exa:pos} has positive pairwise  and conditional covariances except for $\coa_{\alpha\beta|1}=-9/2500$. Thus it does not satisfy \eqref{eq:probs}.
\end{example}

\subsection{Non-identifiable cases}

The above considerations dealt with cases where a given triplet distribution $\dens$ is algebraically identifiable. The final step of the tripod analysis is to regard those distributions that violate the conditions \eqref{eq:alg}. Corollary \ref{coro:notreepair} already discussed the case where one pairwise covariance is zero while the other two are not and we found that they were not tree decomposable. In the following we look at the remaining cases.

\begin{prop}\label{prop:nofullpair}
Assume that a triplet distribution $\dens$ obeys $\coa_{\alpha\beta}\coa_{\alpha\gamma}\coa_{\beta\gamma}=-(\coa_{\alpha\beta\gamma}/2)^2$ but $\coa_{\alpha\beta}\coa_{\alpha\gamma}\coa_{\beta\gamma}\ne0$. Then $\dens$  is not tree decomposable.
\end{prop}

In other words, we found another set of triplet distributions that are not tree decomposable.

\begin{example}
The distribution
\begin{linenomath}
\begin{equation*}
\dens=(16,5,8,15,14,5,2,15)/80
\end{equation*}
\end{linenomath}
yields $\coa_{\alpha\beta}=-1/80,\,\coa_{\alpha\gamma}=1/40$ and $\coa_{\beta\gamma}=1/8$ but $\chi=0$ and hence has no factorisation in the sense of \eqref{eq:tripodtotal}. As in Example \ref{ex:notreepair} we point out here that there seems to be no simple graphical structure which explains the observed covariances adequately. On the other hand, similarly to Example \ref{ex:negative} the simple act of moving $1/80$ from pattern $100$ to pattern $110$ yields algebraic identifiability. This indicates the level of care required when inferring meaning from observed covariances.
\end{example}

Together with Corollary \ref{coro:notreepair} this covers the distributions that are not tree decomposable. The remaining cases are triplet distributions that are tree decomposable but not algebraically identifiable.

\begin{prop}\label{prop:degenerates}
Let $\dens$ be a triplet distribution with $\coa_{\alpha\beta}=0$ and $\coa_{\alpha\gamma}=0$. Then $\dens$ is tree decomposable with infinitely many parameter sets.

The parameter sets are identified by one of the following compositions:
\begin{enumerate}
\item[(i)] $\coa_{\beta\gamma}\ne0$. Then $M^\alpha_{0a}=M^\alpha_{1a}=p_{a\str\str},\,a\in\{0,1\}$, and for any $u,b,c\in\{0,1\}$:
\begin{linenomath}
\begin{equation}
q^\wurzel_1=\frac{p_{\str\str c}-M^\gamma_{1c}}{M^\gamma_{0c}-M^\gamma_{1c}},\quad
M^\beta_{ub}=\frac{p_{\str bc}-p_{\str b\str}M^\gamma_{(1-u)c}}{p_{\str\str c}-M^\gamma_{(1-u)c}}
\end{equation}
\end{linenomath}
with free parameters $M^\gamma_{0c}\ne M^\gamma_{1c}$.
\item[(ii)] $\coa_{\beta\gamma}=0$. Then for all $a,b,c,\in\{0,1\}$ the free parameters can be distributed as follows:
\begin{enumerate}
\item[(a)] $M^\alpha_{0a}=M^\alpha_{1a}=p_{a\str\str},\,M^\beta_{0b}=M^\beta_{1b}=p_{\str b\str}$ and
\begin{linenomath}
\begin{equation}
q^\wurzel_1=\frac{p_{\str\str c}-M^\gamma_{0c}}{M^\gamma_{1c}-M^\gamma_{0c}},
\end{equation}
\end{linenomath}
with free parameters $M^\gamma_{0z}\ne M^\gamma_{1z}$.
\item[(b)] $M^\alpha_{0a}=M^\alpha_{1a}=p_{a\str\str},\,M^\beta_{0b}=M^\beta_{1b}=p_{\str b\str},\,M^\gamma_{0c}=M^\gamma_{1c}=p_{\str\str c}$ with free parameter $q^\wurzel_1$.
\item[(c)] $q^\wurzel_1=0,\,M^\alpha_{0a}=p_{a\str\str},\,M^\beta_{0b}=p_{\str b\str},\,M^\gamma_{0c}=p_{\str\str c}$ with free parameters $M^\alpha_{1a},\,M^\beta_{1b},\,M^\gamma_{1c}$.
\end{enumerate}
\end{enumerate} 
\end{prop}

In other words, the distribution is tree decomposable because process parameters exist but it is not algebraically identifiable because we have no means to recover the true parameters or more precisely, there are infinitely many  parameters that yield the same distribution.

\begin{example}
The triplet distribution
\begin{linenomath}
\begin{equation*}
\dens=(2,2,2,2,2,2,2,2)/16
\end{equation*}
\end{linenomath}
yields complete independence of the leaves $\coa_{\alpha\beta}=\coa_{\alpha\gamma}=\coa_{\beta\gamma}=0$, i.e. the case (ii) in Proposition \ref{prop:degenerates} is to be regarded here. It is not too surprising that such a distribution yields an infinite number of solutions since the state at the root is completely undetermined.
\end{example}

Looking again at the cases listed above, we see that  $X_\alpha$ is not only pairwise independent from $(X_\beta,X_\gamma)$ (induced by $\tau_{\alpha\beta{}}=\tau_{\alpha\gamma}=0$), but even completely independent. Then the multiple solutions come from the fact that we can place the root arbitrarily between $\beta$ and $\gamma$.   

The good news is, that the non-identifiable cases form a small subset among all triplet distributions. In fact:

\begin{prop}\label{prop:lebesgue}
Non-identifiable triplet distributions, i.e. distributions violating the conditions \eqref{eq:alg} form a Lebesgue zero set in the set of all possible triplet distributions.
\end{prop}

This concludes our analysis of the tripod case. We identified the subset of triplet distributions that are uniquely algebraically and stochastically identifiable, and those that are tree decomposable but not algebraically identifiable, or not tree decomposable at all. 

\section{\label{sect.quartet}Extension to quartet trees}

In this section we will explore the implications of extending the results for three taxa to four taxa. For this section we look at the quartet tree $Q=(\vertex,\edge)$ with
\begin{linenomath}
\begin{equation*}
\vertex=\{\wurzel,\psi,\alpha,\beta,\gamma,\delta\},\quad
\edge=\{(\wurzel,\psi),(\wurzel,\alpha),(\wurzel,\beta),(\psi,\gamma),(\psi,\delta)\}.
\end{equation*}
\end{linenomath}
\change{$Q$ is described by the observation that $\alpha,\beta$ are connected to $\wurzel$ and $\gamma,\delta$ are connected to $\psi$. There are two alternative topologies for the leaf set $\{\alpha,\beta,\gamma,\delta\}$, each of which arises by switching the position of $\beta$ with the position of either $\gamma$ or $\delta$ in the edge set $\edge$. The following results are thus valid for all topologies when taking the associated leaf switch into account.}

Fig. \ref{fig:quartet} provides an illustration including the four tripod restrictions $\overline{\tree}=\tree_{\alpha\beta\gamma},\,\widetilde{\tree}=\tree_{\alpha\beta\delta},\,\widehat{\tree}=\tree_{\alpha\gamma\delta}$ and $\check{\tree}=\tree_{\beta\gamma\delta}$.

Regard the quartet distribution  $\bm{\pi}=(\pi_{abcd})_{a,b,c,d\in\{0,1\}}$ describing the joint distribution for $\alpha,\,\beta,\,\gamma$ and $\delta$. If $\bm{\pi}$ is stochastically identifiable and reversible then it can be reconstructed from the \emph{marginalisations} on its four tripods \citep{chang1996}, i.e. computing the parameters for all tripods will immediately return the full process. However, the converse is not necessarily true. As Example \ref{exa:chor} below shows, there are cases where each tripod marginalisation is stochastically identifiable but no quartet tree can be reconstructed.

\begin{figure}[htb]
\begin{center}
\includegraphics[width=\textwidth]{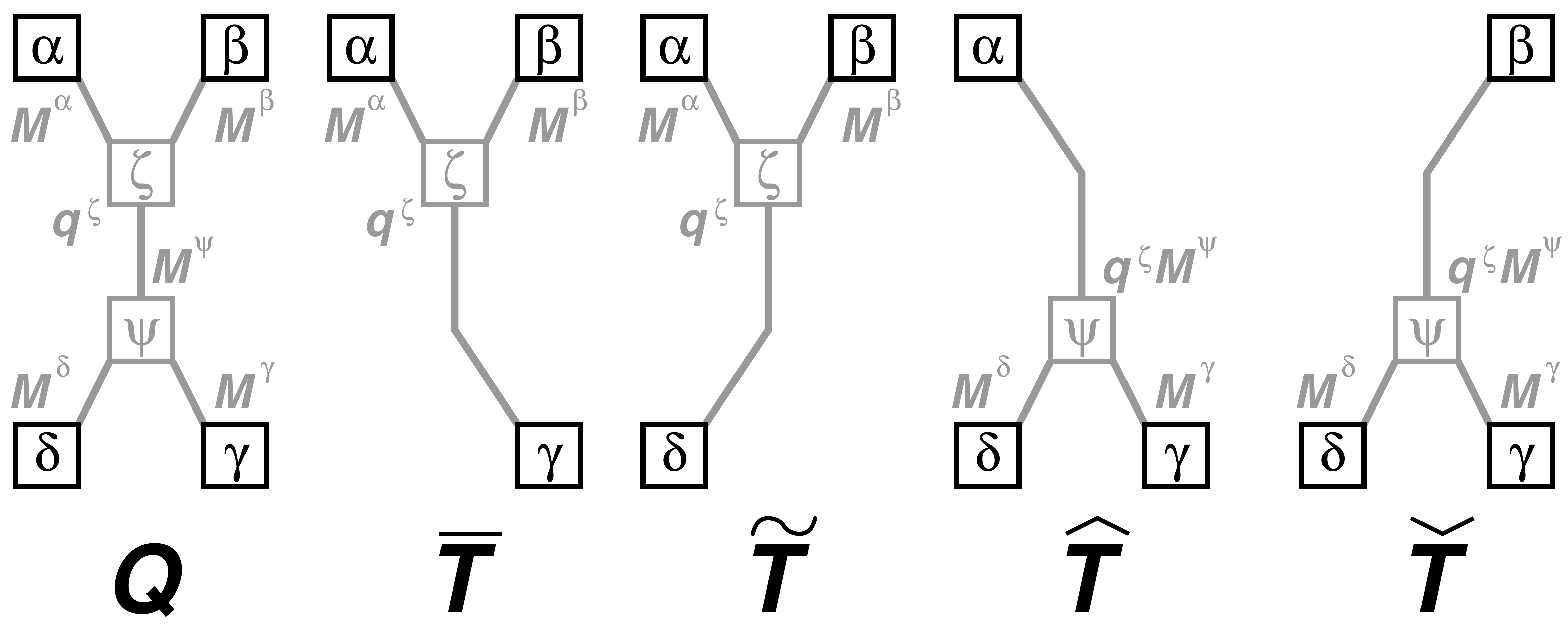}%
\caption{\label{fig:quartet}The quartet tree $Q$ with its tripod restrictions $\overline{\tree},\,\widetilde{\tree},\,\widehat{\tree}$ and $\check{\tree}$. Again, gray lines and vertices indicate the hidden or unknown variables of the approach presented here.}
\end{center}
\end{figure}

\citet{pearl1986} presented an algorithm to reconstruct the topology for an arbitrary number of taxa. Their algorithm employs the condition that tripods that share an interior node in the (unknown) tree topology must have the same marginal distribution at this interior node. Their approach yields an invariant, which for $Q$ amounts to
\begin{linenomath}
\begin{equation}
f_1(\bm{\pi})=\coa_{\alpha\delta}\coa_{\beta\gamma}-\coa_{\alpha\gamma}\coa_{\beta\delta}.
\end{equation}
\end{linenomath}
This invariant is related to the four-point-condition \citep[e.g.,][p. 146]{semple2003} and thus topologically informative, i.e. it is particular to topology $Q$. If a distribution $\bm{\pi}$ is from another tree than $f_1(\bm{\pi})\ne0$.

To reconstruct the process parameters as well, more invariants are needed. In particular, for $\bm{\pi}$ to be algebraically identifiable on $Q$ the parameters obtained from the tripod marginalisations must satisfy the following properties:

\begin{enumerate}
\item The parameters for edges $(\wurzel,\alpha),\,(\wurzel,\beta)$ and $q^\wurzel$ obtained from triplet distributions $\overline{\dens}$ and $\widetilde{\dens}$, respectively, must be equal.
\item The parameters for edges $(\psi,\gamma),\,(\psi,\delta)$ and $q^\psi$ obtained from triplet distributions $\widehat{\dens}$ and $\check{\dens}$, respectively, must be equal.
\item The parameters $\bm{M}^\psi$ for the interior edge $(\wurzel,\psi)$ are obtained from the equations
\begin{linenomath}
\begin{equation}\label{eq:interior_pars}
\begin{split}
\overline{M}^\gamma_{01}&=(1-M^\psi_{01})\widehat{M}^\gamma_{01}+M^\psi_{01}\widehat{M}^\gamma_{11},\\
\overline{M}^\gamma_{11}&=(1-M^\psi_{11})\widehat{M}^\gamma_{01}+M^\psi_{11}\widehat{M}^\gamma_{11}.
\end{split}
\end{equation}
\end{linenomath}
These equations must hold equivalently when $\gamma$ is replaced by $\delta$ and the parameters come from tripod $\widetilde{\tree}$ instead of $\widehat{M}$.
\end{enumerate}

These conditions imply further restrictions on $\bm{\pi}$. An indicator for the minimal number of such conditions is the observation that a quartet distribution $\bm{\pi}$ has 15 degrees of freedom, but there are only 11 process parameters on $Q$, two for each edge and one for the root distribution. Thus we need at least four additional conditions or rather invariants. We will use the above observations to derive an equivalent set of invariants.

\begin{prop}\label{prop:algext}
A quartet distribution $\bm{\pi}$ is algebraically identifiable on $Q$ if its tripod marginalisations satisfy conditions \eqref{eq:alg} and the following invariants vanish on $\bm{\pi}$:
\begin{linenomath}
\begin{align*}
f_0(\bm{\pi})&=\era_{\alpha\beta\gamma\delta}\coa_{\alpha\gamma}-\era_{\alpha\beta\gamma}\era_{\alpha\gamma\delta}+\era_\gamma\era_{\alpha\beta}\era_{\alpha\gamma\delta}+\era_\alpha\era_{\gamma\delta}\era_{\alpha\beta\gamma}-\era_{\alpha\beta}\era_{\alpha\gamma}\era_{\gamma\delta},\\
f_1(\bm{\pi})&=\coa_{\alpha\delta}\coa_{\beta\gamma}-\coa_{\alpha\gamma}\coa_{\beta\delta},\\
f_2(\bm{\pi})&=\coa_{\alpha\gamma}\coa_{\beta\gamma\delta}-\coa_{\beta\gamma}\coa_{\alpha\gamma\delta},\\
f_3(\bm{\pi})&=\coa_{\alpha\gamma}\coa_{\alpha\beta\delta}-\coa_{\alpha\delta}\coa_{\alpha\beta\gamma}.
\end{align*}
\end{linenomath}
The parameters unique up to state flip at the interior nodes are then given by Theorem \ref{thm:alg} and
\begin{linenomath}
\begin{equation}\label{eq:interior}
\begin{split}
M^\psi_{01}&=\frac{1}{2}+\frac{\coa_{\alpha\delta}\coa_{\alpha\beta\gamma}-\coa_{\alpha\beta}\coa_{\alpha\gamma\delta}-\coa_{\alpha\delta}\sqrt{\chi_{\alpha\beta\gamma}}}{2\coa_{\alpha\beta}\sqrt{\chi_{\alpha\gamma\delta}}},\\
M^\psi_{11}&=\frac{1}{2}+\frac{\coa_{\alpha\delta}\coa_{\alpha\beta\gamma}-\coa_{\alpha\beta}\coa_{\alpha\gamma\delta}+\coa_{\alpha\delta}\sqrt{\chi_{\alpha\beta\gamma}}}{2\coa_{\alpha\beta}\sqrt{\chi_{\alpha\gamma\delta}}}.
\end{split}
\end{equation}
\end{linenomath}
\end{prop}

The existence of these invariants means that tree decomposable quartet distributions form a Lebesgue zero set in the set of all quartet distributions for the same reason that the non-identifiable sets are a Lebesgue zero set in the set of all tree decomposable distributions.

Invariant $f_1$ comes from the equality of the marginal distributions at the interior nodes, as proposed by \citet{pearl1986}. Invariants $f_2$ and $f_3$ come from the equality of edge transition matrices. Hence, distributions for which $f_1,\,f_2$ and $f_3$ vanish will uniquely identify topology $Q$. Therefore, $f_1$ to $f_3$ are topologically informative.

However, only distributions for which $f_0$ vanishes will be subject to the inferred parameters. In other words, in the set of zero points for $f_1$ to $f_3$ there is a set of distributions that returns the same set of parameters for $Q$, but only for one of these distributions $f_0$ vanishes. It would be interesting to investigate how this distribution relates to the set it projects from, e.g. if it is related to the possible maximum likelihood optimum.

Despite the fact that $f_1$ to $f_3$ are sufficient to infer a topology, $f_0$ is also topologically informative in that it will not vanish for distributions coming from another tree.

In the case of the CFN model, all triplet covariances vanish. Hence, only invariants $f_0$ and $f_1$ are of interest in that case. Therefore, either invariant is sufficient to identify the associated tree topology.

The parameters for the interior edge do not add more non-identifiable cases. However, as in the tripod case, further conditions are needed to guarantee quartet identifiability.

\begin{prop}\label{prop:stoch}
A quartet distribution is stochastically identifiable if and only if every triplet marginalisations satisfies both Theorem \ref{thm:stoch} and the following inequalities 
\begin{linenomath}
\begin{align}\label{ineq:one}
\coa_{\alpha\delta}\sqrt{\chi_{\alpha\beta\gamma}}-\coa_{\alpha\beta}\sqrt{\chi_{\alpha\gamma\delta}}&\le\coa_{\alpha\beta}\coa_{\alpha\gamma\delta}-\coa_{\alpha\delta}\coa_{\alpha\beta\gamma}\le\coa_{\alpha\beta}\sqrt{\chi_{\alpha\gamma\delta}}-\coa_{\alpha\delta}\sqrt{\chi_{\alpha\beta\gamma}}.
\end{align}
\end{linenomath}
\end{prop}

All other relations are covered due to the fact that the quartet distribution $\dens$ needs to satisfy the invariants $f_0-f_3$. The following example provides a very nice case in which reconstruction is not possible but offers a very interesting challenge.

\begin{example}\label{exa:chor}
\citet{chor2000} discussed several examples of distributions with multiple maxima of the likelihood function. These examples relate to the CFN model, i.e., $p_{abcd}=p_{(1-a)(1-b)(1-c)(1-d)}$ so that the Hadamard approach can be used. Regard the symmetric distribution
\begin{linenomath}
\begin{equation}\label{eq:chor}
\dens=(14,0,0,3,0,2,1,0,0,1,2,0,3,0,0,14)/40.
\end{equation}
\end{linenomath}
Retrieving the statistics yields:
\begin{linenomath}
\begin{align*}
\coa_{\alpha\beta}=7/40=\coa_{\gamma\delta},\quad
\coa_{\alpha\gamma}&=3/20=\coa_{\beta\delta},\quad
\coa_{\alpha\delta}=1/8=\coa_{\beta\gamma},\\
\coa_{\alpha\beta\gamma}=\coa_{\alpha\beta\delta}&=\coa_{\alpha\gamma\delta}=\coa_{\beta\gamma\delta}=0.
\end{align*}
\end{linenomath}
The last equality immediately shows, that the above distribution will trivially satisfy invariants $f_2$ and $f_3$. However, we get $f_1=-11/1600$ and $f_0=-23/375$, i.e. our observations do not come from the quartet tree defined by the bipartition $\alpha\beta|\gamma\delta$. Looking at the alternative invariants for $f_1$, i.e. at
\begin{linenomath}
\begin{align*}
f^{\alpha\delta|\beta\gamma}_1&=\coa_{\alpha\beta}\coa_{\gamma\delta}-\coa_{\alpha\gamma}\coa_{\beta\delta}=13/1600,\\
f^{\alpha\gamma|\beta\delta}_1&=\coa_{\alpha\beta}\coa_{\gamma\delta}-\coa_{\alpha\delta}\coa_{\beta\gamma}=3/200,
\end{align*}
\end{linenomath}
we see that this distribution comes from none of the available quartet trees.

Nevertheless, we shall have a look at the parameters. Note that the symmetry of the distribution $\dens$ implies $M^\alpha_{01}=1-M^\alpha_{11}=:\overline{M}_\alpha$. Looking at the numerical values for the parameters for every tripod tree we find surprising similarities:

\begin{table}[h]
\begin{center}
\begin{tabular}{r|c|c|c|c|c}
triplet&$M^\alpha_{01}$&$M_\beta$&$M_\gamma$&$M_\delta$&$q_\wurzel$\\\hline
$\alpha\beta\gamma$&0.0417424&0.118119&0.172673&0&0.5\\
$\alpha\beta\delta$&0.118119&0.0417424&0&0.172673&0.5\\\hline
$\alpha\gamma\delta$&0.172673&0&0.0417424&0.118119&0.5\\
$\beta\gamma\delta$&0&0.172673&0.118119&0.0417424&0.5
\end{tabular}
\caption{\label{tab:values}The parameters for each triplet.}
\end{center}
\end{table}

These parameters permit us to infer parameters $M_\wurzel=1/14$ and $M_\psi=1/7$ such that e.g. the parameters for $\alpha$ on the tripod trees $\alpha\beta\delta$ and $\alpha\gamma\delta$ can be obtained from the parameter for tripod tree $\alpha\beta\gamma$ by
\begin{linenomath}
\begin{align*}
\widetilde{M}_\alpha&=M_\wurzel(1-\overline{M}_\alpha)+(1-M_\wurzel)\overline{M}_\alpha,\quad
\widehat{M}_\alpha=M_\psi(1-\overline{M}_\alpha)+(1-M_\psi)\overline{M}_\alpha,
\end{align*}
\end{linenomath}
with analogue assignments for the other leaves. These computations can be visualised by the network in Fig. \ref{fig:network}. The assignment of probabilities for each split permits to justify the observations for each of the four tripod trees. However, the visualisation is misleading because the factorisation of the system does not follow the edges in the network \citep[e.g.,][]{churchill1993,strimmer2000,bryant2005}. 

\begin{figure}[htb]
\begin{center}
\includegraphics[width=\textwidth]{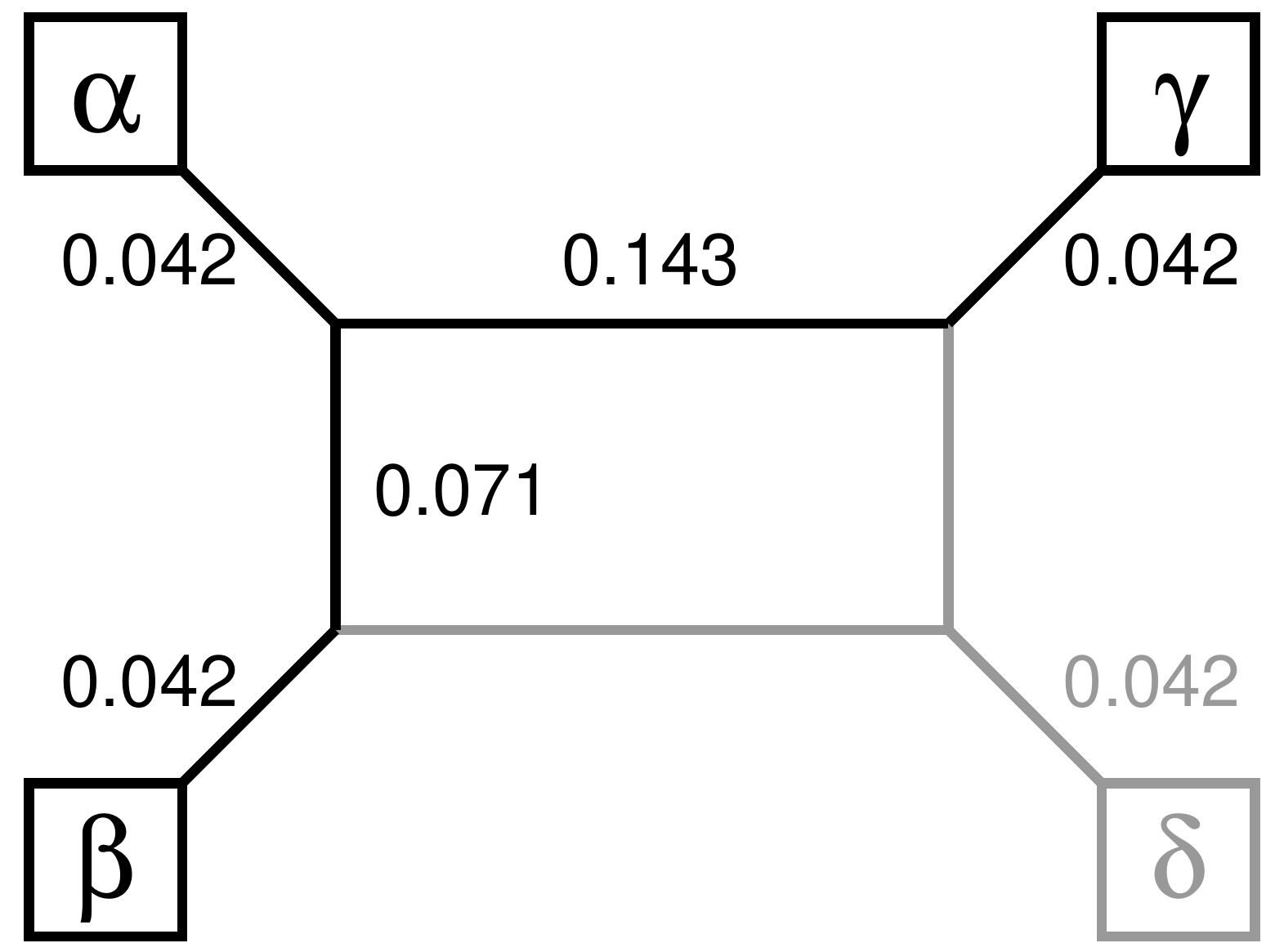}%
\caption{\label{fig:network}Assignment of mutation probability from the symmetric distribution in \eqref{eq:chor}. The black lines indicate the triplet $\alpha\beta\gamma$. Assigned branch lengths are rounded values.}
\end{center}
\end{figure}

\end{example}

\section{The connection with Markov invariants}\label{sect:markovinv}

This section investigates the connection between the work presented here and the concept of Markov invariants as coined by \citet{sumner2008}. To show these relations we will look back at our covariances and investigate their relationship with the parameters.

Following \citet{allman2008} one can write the three-way-probabilities as a $2\times2\times2$ tensor $\bm{P}^{\alpha\beta\gamma}$ such that
\begin{linenomath}
\begin{equation*}
\bm{P}^{\alpha\beta|0}=\begin{pmatrix}
p_{000}&p_{010}\\p_{100}&p_{110}
\end{pmatrix},\quad
\bm{P}^{\alpha\beta|1}=\begin{pmatrix}
p_{001}&p_{011}\\p_{101}&p_{111}
\end{pmatrix}.
\end{equation*}
\end{linenomath}
With this as a basis we easily infer our pairwise covariances in terms of determinants of dimensional restrictions of $\bm{P}^{\alpha\beta\gamma}$. E.g., a marginalisation over $\gamma$ corresponds to $\bm{P}^{\alpha\beta\str}=\bm{P}^{\alpha\beta|0}+\bm{P}^{\alpha\beta|1}$. The determinant of this matrix then corresponds to
\begin{linenomath}
\begin{align*}
\det \bm{P}^{\alpha\beta\str}&=p_{00\str}p_{11\str}-p_{01\str}p_{10\str}\\
&=p_{11\str}(p_{00\str}+p_{11\str}+p_{01\str}+p_{10\str})-(p_{10\str}+p_{11\str})(p_{01\str}+p_{11\str})\\
&=p_{11\str}-p_{1\str\str}p_{\str1\str}=\coa_{\alpha\beta}.
\end{align*}
\end{linenomath}
Thus, we have invariably obtained an alternative way to compute the covariances. In a similar fashion, if we take the determinant of the conditional kernels $\bm{P}^{\alpha\beta|c},\,c\in\{0,1\}$, we arrive at the (not normalised) conditional covariance $\coa_{\alpha\beta|c}$:
\begin{linenomath}
\begin{align*}
\det \bm{P}^{\alpha\beta|c}&=p_{00c}p_{11c}-p_{01c}p_{10c}\\
&=p_{11c}(p_{\str\str c}-p_{01c}-p_{10c}-p_{11c})-(p_{\str1c}-p_{11c})(p_{1\str c}-p_{11c})\\
&=p_{11c}p_{\str\str c}-p_{\str1c}p_{1\str c}=\coa_{\alpha\beta|c}.
\end{align*}
\end{linenomath}
It must be noted that the determinant has been used earlier in connection with LogDet families \citep[e.g.,][]{steel1994}. In order to relate these findings to the process parameter, let us denote by $\bm{\Pi}=\text{diag}(\bm{q}^\wurzel)$ the diagonal matrix of the marginal distribution at the root, and with
\begin{linenomath}
\begin{align*}
\bm{M}^\alpha=\begin{pmatrix}
M^\alpha_{00}&M^\alpha_{01}\\M^\alpha_{10}&M^\alpha_{11}.
\end{pmatrix}
\end{align*}
\end{linenomath}
the transition matrix for leaf $\alpha$. Then the marginalisation of Equation \eqref{eq:leaftotal} can be written as
\begin{linenomath}
\begin{equation}\label{eq:pairmat}
\bm{P}^{\alpha\beta\str}=(\bm{M}^\alpha)^{\text{T}}\overline{\bm{\Pi}}\bm{M}^\beta,
\end{equation}
\end{linenomath}
where $\overline{\bm{\Pi}}$ is the marginal distribution at the most recent common ancestor of $\alpha$ and $\beta$. If $\mathcal{E}_{\alpha\beta}$ is defined as the set of edges connecting the root of the tree and the most recent common ancestor of $\alpha$ and $\beta$ then we compute $\overline{\bm{\Pi}}$ by
\begin{linenomath}
\begin{equation*}
\overline{\bm{\Pi}}=\bm{\Pi}\prod_{e\in\mathcal{E}_{\alpha\beta}}\bm{M}^e.
\end{equation*}
\end{linenomath}
If we take the determinant on both sides of Eq. \eqref{eq:pairmat} we get
\begin{linenomath}
\begin{equation*}
\det\bm{P}^{\alpha\beta\str}=\det\bm{M}^\alpha\det\bm{M}^\beta\det\bm{\Pi}\prod_{e\in\mathcal{E}_{\alpha\beta}}\det\bm{M}^e.
\end{equation*}
\end{linenomath}
We further observe that the determinant in the two-state-case is equal to
\begin{linenomath}
\begin{equation*}
\det\bm{M}^\alpha=1-M^\alpha_{00}-M^\alpha_{11}=-(M^\alpha_{11}-M^\alpha_{01}).
\end{equation*}
\end{linenomath}
Going back to a tripod tree under the two-state-model this yields the relation
\begin{linenomath}
\begin{equation}\label{eq:treepaircov}
\coa_{\alpha\beta}=(M^\alpha_{11}-M^\alpha_{01})(M^\beta_{11}-M^\beta_{01})q^\wurzel_1(1-q^\wurzel_1).
\end{equation}
\end{linenomath}
This relation has been observed in \citet{steel1994} and forms the basis for LogDet inference. The covariances $\coa_{\alpha\beta}$ also form the simplest form of \emph{Markov invariants}. \citet{sumner2008} define these terms in general by:
\begin{linenomath}
\begin{equation}\label{eq:markovinvariant}
f(\dens)=g(\widehat{\dens})\prod_{e\in\edge}(\det\bm{M}^e)^{k_e},
\end{equation}
\end{linenomath}
with $k_e\in\mathbb{Z}$ denoting the exponent for edge $e\in\edge$. The term $g(\widehat{\dens})$ describes a function depicting the relationship of a reduced structure in the tree. \citet{sumner2008} give one example of such a reduced structure as the tree for which the pendant edges have been reduced to length zero. In the case of the tripod tree this reduced structure corresponds to the interior node $\wurzel$, and hence the distribution $\widehat{\dens}$ is equivalent to $\bm{q}^\wurzel$ only. In this setting, Markov invariants are one-dimensional ``representations'' of the stochastic models used for inference, such that the complex structure of these models is retained \citep{sumner2009}.

In our framework, we rediscover more Markov invariants of type \eqref{eq:markovinvariant} when investigating how the remaining covariances are related to the process parameters under the tripod equations \eqref{eq:tripodtotal}. In fact, we find:
\begin{linenomath}
\begin{align}
\label{eq:treetriplecov}
\coa_{\alpha\beta\gamma}&=(M^\alpha_{11}-M^\alpha_{01})(M^\beta_{11}-M^\beta_{01})(M^\gamma_{11}-M^\gamma_{01})q^\wurzel_1(1-q^\wurzel_1)(1-2q^\wurzel_1),\\
\label{eq:product}
\coa_{\alpha\beta}\coa_{\alpha\gamma}\coa_{\beta\gamma}&=
(M^\alpha_{11}-M^\alpha_{01})^2(M^\beta_{11}-M^\beta_{01})^2(M^\gamma_{11}-M^\gamma_{01})^2(q^\wurzel_1)^3(1-q^\wurzel_1)^3,\\
\label{eq:chimarkov}
\chi&=(M^\alpha_{11}-M^\alpha_{01})^2(M^\beta_{11}-M^\beta_{01})^2(M^\gamma_{11}-M^\gamma_{01})^2(q^\wurzel_1)^2(1-q^\wurzel_1)^2.
\end{align}
\end{linenomath}
with equivalent terms for the other covariances. These equivalences permit a different way to prove Theorem \ref{thm:alg} from the one we present in \ref{app:proofs}.

It should be noted that our interpretation of the above Markov invariants as covariances only works for the two state model. On the other hand, the form of the Markov invariants stays valid, even though they might not be as immediately apparent from the model as in the cases discussed here. However, in the case of the two-state model using the notion of covariance permits a good interpretation of the findings.

We observe for the (not normalised) conditional covariances
\begin{linenomath}
\begin{align}
\label{eq:treecondpaircov}
\coa_{\alpha\beta|c}&=(M^\alpha_{11}-M^\alpha_{01})(M^\beta_{11}-M^\beta_{01})M^\gamma_{0c}M^\gamma_{1c}q^\wurzel_1(1-q^\wurzel_1),
\end{align}
\end{linenomath}
i.e., the transition matrix for leaf $\gamma$ shall be included into the term $g(\widehat{\dens})$ for \eqref{eq:markovinvariant} to be valid. On the other hand, remember that we did not use these covariances to solve the tripod equations \eqref{eq:tripodtotal}. We need them only to formulate the positivity constraints in Theorem \ref{thm:stoch}. This property is beyond the purely algebraic framework.

In summary, Markov invariants are very useful when investigating properties of and conditions on leaf distributions $\dens$. Especially, they explore the relationship of process parameters and leaf distribution such that phylogenetic invariants like $f_1$ to $f_3$ from Proposition \ref{prop:algext} can be easily extracted. We will employ these relationships to prove the results of Section \ref{sect.tree}.

\section{Discussion}\label{sect:discussion}

\paragraph{Acknowledgements}
We thank  Elizabeth S.  Allman and John A. Rhodes for stimulating the finalization of the manuscript as well as for sharing their thoughts on this subject with us. Further, we owe much to the discussions with David Bryant, Mike Steel, and Arndt von Haeseler. Suggestions from Jessica Leigh and three anonymous referees are greatly appreciated.


\appendix

\section{Proofs}\label{app:proofs}

\begin{proof}[Proof of Lemma \ref{lemma:flip}]
A state flip replaces the probabilities at leaf $\alpha$ implies a ``new'' distribution $\widehat{\dens}$ with $\widehat{p}_{abc}=p_{(1-a)bc},\,a,b,c\in\{0,1\}$. This has the following implications to the covariances.
\begin{linenomath}
\begin{align*}
\coa_{\alpha\beta}&=p_{11\str}-p_{1\str\str}p_{\str1\str}=(p_{\str1\str}-p_{01\str})-p_{1\str\str}p_{\str1\str}\\
&=-p_{01\str}+p_{\str1\str}(1-p_{1\str\str})=-(p_{01\str}-p_{0\str\str}p_{\str1\str})\\
&=-(\widehat{p}_{11\str}-\widehat{p}_{1\str\str}\widehat{p}_{\str1\str})=-\widehat{\coa}_{\alpha\beta}.
\end{align*}
\end{linenomath}
and analogously $\widehat{\coa}_{\alpha\gamma}=-\coa_{\alpha\gamma}$ and $\widehat{\coa}_{\beta\gamma}=\coa_{\beta\gamma}$. Thus, if $\coa_{\alpha\beta}$ and $\coa_{\alpha\gamma}$ are smaller than zero, then a state flip produces positive covariances and the sign for the overall product remains the same.
\end{proof}

\begin{proof}[Proof of Lemma \ref{lem:indi}]
Using the Markov invariants from Section \ref{sect:markovinv} we immediately see, that if $\coa_{\alpha\beta}=0$ due to $M^\alpha_{01}-M^\alpha_{11}=0$ then also $\coa_{\alpha\gamma}=0$ and $\coa_{\alpha\beta\gamma}=0$. If $q^\wurzel_1\in\{0,1\}$ then all four covariances are zero.

For point 2 regard \eqref{eq:product}. But this term will be non-negative as long as $q^\wurzel_1$ is a probability, which is a model condition. This completes the proof.
\end{proof}

\begin{proof}[Proof of Corollary \ref{coro:positiveflipping}]
Select one leaf $\alpha\in L$ and define $L_0=\{\beta:\coa_{\alpha\beta}<0\}$. Flipping the states in $L_0$ gives us $\coa_{\alpha\beta}>0$ for all $\beta\in L,\beta\ne\alpha$ by Lemma \ref{lemma:flip}. Fix now $\beta\ne \beta'\in L\setminus\{\alpha\}$. Then $\alpha,\beta,\beta'$, together with the root $\wurzel$ of the tree, define uniquely a tripod tree and the restriction of $\dens$ to $\alpha,\beta,\beta'$ must  obey the tripod equations. Using  Lemma \ref{lem:indi}(2), on this tripod tree  shows now that $\coa_{\alpha\beta}\coa_{\alpha\beta'}\coa_{\beta\beta'}>0$. This implies that $\coa_{\beta\beta'}$ is positive, too.  
\end{proof}

\begin{proof}[Proof of Corollary \ref{coro:notreepair}]
A triplet distribution $\dens$ for which only one covariance is zero does not satisfy Lemma \ref{lem:indi}(1) and hence is not tripod decomposable. Further, by looking at \eqref{eq:treepaircov} we see that there is also no real- or complex-valued parameter set that would yield only one zero covariance. Hence, such a triplet distribution would also not be algebraically decomposable.
\end{proof}

\begin{proof}[Proof of Lemma \ref{lem:stateflip}]
We insert the refined parameters into the tripod equations to get:
\begin{linenomath}
\begin{align*}
p_{abc}&=q^\wurzel_1M^\alpha_{1a}M^\beta_{1b}M^\gamma_{1c}+(1-q^\wurzel_1)M^\alpha_{0a}M^\beta_{0b}M^\gamma_{0c}\\
&=\widehat{q}^\wurzel_0\widehat{M}^\alpha_{0a}\widehat{M}^\beta_{0b}\widehat{M}^\gamma_{0c}+(1-\widehat{q}^\wurzel_0)\widehat{M}^\alpha_{1a}\widehat{M}^\beta_{1b}\widehat{M}^\gamma_{1c}\\
&=(1-\widehat{q}^\wurzel_1)\widehat{M}^\alpha_{0a}\widehat{M}^\beta_{0b}\widehat{M}^\gamma_{0c}+\widehat{q}^\wurzel_1\widehat{M}^\alpha_{1a}\widehat{M}^\beta_{1b}\widehat{M}^\gamma_{1c},
\end{align*}
\end{linenomath}
i.e. the tripod equations are recovered with flipped parameters. This completes the proof.
\end{proof}

\begin{proof}[Proof of Theorem \ref{thm:alg}]
We derive the parameters from the tripod equations. As mentioned in Section \ref{sect:statbin} there is a linear relationship between $\dens$ and its marginalisations. Thus, finding a solution for the tripod equations is equivalent to finding the solution for the following set of equations
\begin{linenomath}
\begin{align}
\label{eq:eabc}\era_{\alpha\beta\gamma}&=q^\wurzel_1M^\alpha_{11}M^\beta_{11}M^\gamma_{11}+(1-q^\wurzel_1)M^\alpha_{01}M^\beta_{01}M^\gamma_{01},\\
\label{eq:eab}\era_{\alpha\beta}&=q^\wurzel_1M^\alpha_{11}M^\beta_{11}+(1-q^\wurzel_1)M^\alpha_{01}M^\beta_{01},\\
\label{eq:eac}\era_{\alpha\gamma}&=q^\wurzel_1M^\alpha_{11}M^\gamma_{11}+(1-q^\wurzel_1)M^\alpha_{01}M^\gamma_{01},\\
\label{eq:ebc}\era_{\beta\gamma}&=q^\wurzel_1M^\beta_{11}M^\gamma_{11}+(1-q^\wurzel_1)M^\beta_{01}M^\gamma_{01},\\
\label{eq:ealpha}\era_\alpha&=q^\wurzel_1M^\alpha_{11}+(1-q^\wurzel_1)M^\alpha_{01},\\
\label{eq:ebeta}\era_\beta&=q^\wurzel_1M^\beta_{11}+(1-q^\wurzel_1)M^\beta_{01},\\
\label{eq:egamma}\era_\gamma&=q^\wurzel_1M^\gamma_{11}+(1-q^\wurzel_1)M^\gamma_{01}.
\end{align}
\end{linenomath}
Equations \eqref{eq:ealpha}-\eqref{eq:egamma} yield
\begin{linenomath}
\begin{equation}\label{eq:step1}
(1-q^\wurzel_1)M^\alpha_{01}=\era_\alpha-q^\wurzel_1M^\alpha_{11},\quad
(1-q^\wurzel_1)M^\beta_{01}=\era_\beta-q^\wurzel_1M^\beta_{11},\quad
(1-q^\wurzel_1)M^\gamma_{01}=\era_\gamma-q^\wurzel_1M^\gamma_{11}.
\end{equation}
\end{linenomath}
Inserting \eqref{eq:step1} into \eqref{eq:eab} returns
\begin{linenomath}
\begin{align*}
(1-q^\wurzel_1)\era_{\alpha\beta}&=q^\wurzel_1(1-q^\wurzel_1)M^\alpha_{11}M^\beta_{11}+(\era_\alpha-q^\wurzel_1M^\alpha_{11})(\era_\beta-q^\wurzel_1M^\beta_{11})\\
&=q^\wurzel_1M^\alpha_{11}M^\beta_{11}+\era_\alpha\era_\beta-q^\wurzel_1(\era_\alpha M^\beta_{11}+\era_\beta M^\alpha_{11}),
\end{align*}
\end{linenomath}
and in consequence
\begin{linenomath}
\begin{align}
\label{eq:step2a}q^\wurzel_1M^\beta_{11}(M^\alpha_{11}-\era_\alpha)&=\coa_{\alpha\beta}+q^\wurzel_1(\era_\beta M^\alpha_{11}-\era_{\alpha\beta}),\\
\label{eq:step2b}q^\wurzel_1M^\gamma_{11}(M^\alpha_{11}-\era_\alpha)&=\coa_{\alpha\gamma}+q^\wurzel_1(\era_\gamma M^\alpha_{11}-\era_{\alpha\gamma}).
\end{align}
\end{linenomath}
We insert \eqref{eq:step2a}-\eqref{eq:step2b} back into \eqref{eq:step1}
\begin{linenomath}
\begin{align*}
(1-q^\wurzel_1)M^\beta_{01}(M^\alpha_{11}-\era_\alpha)&=\era_\beta(M^\alpha_{11}-\era_\alpha)-\coa_{\alpha\beta}-q^\wurzel_1(\era_\beta M^\alpha_{11}-\era_{\alpha\beta})\\
&=(1-q^\wurzel_1)(\era_\beta M^\alpha_{11}-\era_{\alpha\beta}).
\end{align*}
\end{linenomath}
In the case of $q^\wurzel_1=1$ we get from \eqref{eq:ealpha} and \eqref{eq:eab} that $M^\alpha_{11}=\era_\alpha$ and $\era_{\alpha\beta}=\era_\alpha\era_\beta$. Hence, we remove $1-q^\wurzel_1$ from the above equation without destroying equality. Thus, we get
\begin{linenomath}
\begin{align}
\label{eq:step3a}M^\beta_{01}(M^\alpha_{11}-\era_\alpha)&=\era_\beta M^\alpha_{11}-\era_{\alpha\beta},\\
\label{eq:step3b}M^\gamma_{01}(M^\alpha_{11}-\era_\alpha)&=\era_\gamma M^\alpha_{11}-\era_{\alpha\gamma}.
\end{align}
\end{linenomath}
We insert \eqref{eq:step1} in \eqref{eq:eabc} to get
\begin{linenomath}
\begin{equation*}
M^\alpha_{11}\era_{\beta\gamma}-\era_{\alpha\beta\gamma}=M^\beta_{01}M^\gamma_{01}(M^\alpha_{11}-\era_\alpha).
\end{equation*}
\end{linenomath}
Applying \eqref{eq:step3a} and \eqref{eq:step3b} to this gives us
\begin{linenomath}
\begin{align*}
0&=(M^\alpha_{11}\era_{\beta\gamma}-\era_{\alpha\beta\gamma})(M^\alpha_{11}-\era_\alpha)-(\era_\beta M^\alpha_{11}-\era_{\alpha\beta})(\era_\gamma M^\alpha_{11}-\era_{\alpha\gamma})\\
&=(M^\alpha_{11})^2\coa_{\beta\gamma}-M^\alpha_{11}(\coa_{\alpha\beta\gamma}+2\era_\alpha\coa_{\beta\gamma})+\era_{\alpha\beta\gamma}\era_\alpha-\era_{\alpha\beta}\era_{\alpha\gamma}.
\end{align*}
\end{linenomath}
We can apply the solution formula for quadratic equations provided $\coa_{\beta\gamma}\ne0$, i.e. our condition \eqref{eq:alg} is satisfied. In that case we get
\begin{linenomath}
\begin{align}
\notag(M^\alpha_{11})_{\pm}&=\frac{\coa_{\alpha\beta\gamma}+2\era_\alpha\coa_{\beta\gamma}}{2\coa_{\beta\gamma}}\pm\frac{\sqrt{(\coa_{\alpha\beta\gamma}+2\era_\alpha\coa_{\beta\gamma})^2-4(\era_{\alpha\beta\gamma}\era_\alpha-\era_{\alpha\beta}\era_{\alpha\gamma})\coa_{\beta\gamma}}}{2\coa_{\beta\gamma}}\\
\label{eq:ma11}&=\era_\alpha+\frac{\coa_{\alpha\beta\gamma}\pm\sqrt{\coa_{\alpha\beta\gamma}^2+4\coa_{\alpha\beta}\coa_{\alpha\gamma}\coa_{\beta\gamma}}}{2\coa_{\beta\gamma}}.
\end{align}
\end{linenomath}
Thus we have established the term for $M^\alpha_{11}$. The next step is to derive $q^\wurzel_1$. We insert \eqref{eq:step2a}-\eqref{eq:step3b} into \eqref{eq:ebc} and get
\begin{linenomath}
\begin{align*}
q^\wurzel_1(M^\alpha_{11}-\era_\alpha)^2\era_{\beta\gamma}&=
(\coa_{\alpha\beta}+q^\wurzel_1(\era_\beta M^\alpha_{11}-\era_{\alpha\beta}))(\coa_{\alpha\gamma}+q^\wurzel_1(\era_\gamma M^\alpha_{11}-\era_{\alpha\gamma}))\\
&+q^\wurzel_1(1-q^\wurzel_1)(\era_\beta M^\alpha_{11}-\era_{\alpha\beta})(\era_\gamma M^\alpha_{11}-\era_{\alpha\gamma})\\
&=(1-q^\wurzel_1)\coa_{\alpha\beta}\coa_{\alpha\gamma}+q^\wurzel_1\era_\beta\era_\gamma(M^\alpha_{11}-\era_\alpha)^2
\end{align*}
\end{linenomath}
and hence we get the quadratic relation
\begin{linenomath}
\begin{equation}
0=(1-q^\wurzel_1)\coa_{\alpha\beta}\coa_{\alpha\gamma}-q^\wurzel_1\coa_{\beta\gamma}(M^\alpha_{11}-\era_\alpha)^2
\end{equation}
\end{linenomath}
We insert \eqref{eq:ma11} and get
\begin{linenomath}
\begin{align*}
\coa_{\alpha\beta}\coa_{\alpha\gamma}&=q^\wurzel_1\big(\coa_{\alpha\beta}\coa_{\alpha\gamma}+\coa_{\beta\gamma}(M^\alpha_{11}-\era_\alpha)^2\big),\\
4\coa_{\alpha\beta}\coa_{\alpha\gamma}\coa_{\beta\gamma}&=q^\wurzel_1\bigg(4\coa_{\alpha\beta}\coa_{\alpha\gamma}\coa_{\beta\gamma}+\big(\coa_{\alpha\beta\gamma}+\sqrt{\chi}\big)^2\bigg),\\
4\coa_{\alpha\beta}\coa_{\alpha\gamma}\coa_{\beta\gamma}&=2q^\wurzel_1\sqrt{\chi}\big(\sqrt{\chi}+\coa_{\alpha\beta\gamma}\big).
\end{align*}
\end{linenomath}
We use the equality 
\begin{linenomath}
\begin{align*}
4\coa_{\alpha\beta}\coa_{\alpha\gamma}\coa_{\beta\gamma}=\chi-\coa_{\alpha\beta\gamma}^2=(\sqrt{\chi}+\coa_{\alpha\beta\gamma})(\sqrt{\chi}-\coa_{\alpha\beta\gamma})
\end{align*}
\end{linenomath}
and the observation that $\sqrt{\chi}(\sqrt{\chi}-\coa_{\alpha\beta\gamma})=0$ if and only if the conditions in \eqref{eq:alg} are violated to get
\begin{linenomath}
\begin{equation}\label{eq:qwurzel1}
q^\wurzel_1=\frac{\sqrt{\chi}-\coa_{\alpha\beta\gamma}}{2\sqrt{\chi}}=\frac{1}{2}-\frac{\coa_{\alpha\beta\gamma}}{2\sqrt{\chi}},
\end{equation}
\end{linenomath}
thus inferring the proposed term for $q^\wurzel_1$. Next we infer the term for $M^\alpha_{01}$. To this end we insert \eqref{eq:ma11} and \eqref{eq:qwurzel1} into \eqref{eq:step1}:
\begin{linenomath}
\begin{align*}
-q^\wurzel_1(M^\alpha_{11}-\era_\alpha)&=(1-q^\wurzel_1)(M^\alpha_{01}-\era_\alpha),\\
(\coa_{\alpha\beta\gamma}-\sqrt{\chi})(\coa_{\alpha\beta\gamma}+\sqrt{\chi})&=2\coa_{\beta\gamma}(\coa_{\alpha\beta\gamma}+\sqrt{\chi})(M^\alpha_{01}-\era_\alpha),\\
M^\alpha_{01}&=\era_\alpha+\frac{\coa_{\alpha\beta\gamma}-\sqrt{\chi}}{2\coa_{\beta\gamma}},
\end{align*}
\end{linenomath}
thus inferring the proposed term. The remaining terms are inferred analogously. This completes the proof.
\end{proof}

\begin{proof}[Proof of Theorem \ref{thm:stoch}]
We bound the parameters from \eqref{eq:sol} between 0 and 1:
\begin{linenomath}
\begin{align*}
0&\le\frac{1}{2}-\frac{\coa_{\alpha\beta\gamma}}{2\sqrt{\chi}}\le1,\\
-\sqrt{\chi}&\le\coa_{\alpha\beta\gamma}\le\sqrt{\chi},\\
0&\le\coa_{\alpha\beta}\coa_{\alpha\gamma}\coa_{\beta\gamma}.
\end{align*}
\end{linenomath}
With \eqref{eq:alg} this yields positivity for the unconditional covariances. Next we look at $M^\alpha_{01}$ and $M^\alpha_{11}$:
\begin{linenomath}
\begin{align*}
0&\le\era_\alpha+\frac{\coa_{\alpha\beta\gamma}-\sqrt{\chi}}{2\coa_{\beta\gamma}}\le1,\\
-2\era_\alpha\coa_{\beta\gamma}&\le\coa_{\alpha\beta\gamma}-\sqrt{\chi}\le2(1-\era_\alpha)\coa_{\beta\gamma},\\
\coa_{\alpha\beta\gamma}-2(1-\era_\alpha)\coa_{\beta\gamma}&\le\sqrt{\chi}\le\coa_{\alpha\beta\gamma}+2\era_\alpha\coa_{\beta\gamma}
\end{align*}
\end{linenomath}
and
\begin{linenomath}
\begin{align*}
0&\le\era_\alpha+\frac{\coa_{\alpha\beta\gamma}+\sqrt{\chi}}{2\coa_{\beta\gamma}}\le1,\\
-2\era_\alpha\coa_{\beta\gamma}&\le\coa_{\alpha\beta\gamma}+\sqrt{\chi}\le2(1-\era_\alpha)\coa_{\beta\gamma},\\
-(2\era_\alpha\coa_{\beta\gamma}+\coa_{\alpha\beta\gamma})&\le\sqrt{\chi}\le2(1-\era_\alpha)\coa_{\beta\gamma}-\coa_{\alpha\beta\gamma}.
\end{align*}
\end{linenomath}
Squaring both inequalities reduces the four inequalities to the following two:
\begin{linenomath}
\begin{align}
\label{eq:ineq1}\coa_{\alpha\beta\gamma}^2+4\coa_{\alpha\beta}\coa_{\alpha\gamma}\coa_{\beta\gamma}&\le(2\era_\alpha\coa_{\beta\gamma}+\coa_{\alpha\beta\gamma})^2,\\
\label{eq:ineq2}\coa_{\alpha\beta\gamma}^2+4\coa_{\alpha\beta}\coa_{\alpha\gamma}\coa_{\beta\gamma}&\le(2(1-\era_\alpha)\coa_{\beta\gamma}-\coa_{\alpha\beta\gamma})^2.
\end{align}
\end{linenomath}
We look first at inequality \eqref{eq:ineq1} and get
\begin{linenomath}
\begin{align*}
\coa_{\alpha\beta}\coa_{\alpha\gamma}\coa_{\beta\gamma}&\le \era_\alpha^2\coa_{\beta\gamma}^2+\era_\alpha\coa_{\beta\gamma}\coa_{\alpha\beta\gamma},\\
0&\le\era_\alpha(\era_\alpha\coa_{\beta\gamma}+\coa_{\alpha\beta\gamma})-\coa_{\alpha\beta}\coa_{\alpha\gamma},\\
0&\le\era_\alpha\era_{\alpha\beta\gamma}-\era_{\alpha\beta}\era_{\alpha\gamma}=\coa_{\beta\gamma|1}.
\end{align*}
\end{linenomath}
Set $\widehat{\era}_\alpha:=(1-\era_\alpha)=p_{0\str\str}$ and look at \eqref{eq:ineq2}:
\begin{linenomath}
\begin{align*}
\coa_{\alpha\beta}\coa_{\alpha\gamma}\coa_{\beta\gamma}&\le\widehat{\era}_\alpha^2\coa_{\beta\gamma}^2-\widehat{\era}_\alpha\coa_{\beta\gamma}\coa_{\alpha\beta\gamma},\\
0&\le\widehat{\era}_\alpha(\widehat{\era}_\alpha\coa_{\beta\gamma}-\coa_{\alpha\beta\gamma})-\coa_{\alpha\beta}\coa_{\alpha\gamma},\\
0&\le p_{000}p_{011}-p_{001}p_{010}=\coa_{\beta\gamma|0}.
\end{align*}
\end{linenomath}
Hence, we have derived the proposed inequalities.
\end{proof}

\begin{proof}[Proof of Proposition \ref{prop:nofullpair}]
The tripod equations \eqref{eq:tripodtotal} imply:
\begin{linenomath}
\begin{equation*}
\chi=\coa_{\alpha\beta\gamma}^2+4\coa_{\alpha\beta}\coa_{\alpha\gamma}\coa_{\beta\gamma}=
(M^\alpha_{11}-M^\alpha_{01})^2(M^\beta_{11}-M^\beta_{01})^2(M^\gamma_{11}-M^\gamma_{01})^2(1-q^\wurzel_1)^2(q^\wurzel_1)^2
\end{equation*}
\end{linenomath}
Together with \eqref{eq:treetriplecov} and \eqref{eq:treepaircov} we see that there is no set of real or complex parameters such that $\chi=0$ but $\coa_{\alpha\beta}\coa_{\alpha\gamma}\coa_{\beta\gamma}\ne0$.
\end{proof}

\begin{proof}[Proof of Proposition \ref{prop:degenerates}]
The cases are easily verified by looking at Equation \eqref{eq:treepaircov} and inserting the selected parameters back into \eqref{eq:tripodtotal}.
\end{proof}

\begin{proof}[Proof of Proposition \ref{prop:lebesgue}]
The function $\chi:\,\mathbb{C}^8\to\mathbb{C}$ is a nonconstant polynomial mapping. Thus  the set $\{\dens\in\mathbb{R}^8:\,\chi(\dens)=0\}$ is a Lebesgue zero set. The same holds for the set 
\begin{linenomath}
\begin{align*}
\{\dens\in\mathbb{R}^8:\,\coa_{\alpha\beta}(\dens)=0\text{ or }\coa_{\alpha\gamma}(\dens)=0\text{ or }\coa_{\beta\gamma}(\dens)=0\}.
\end{align*}
\end{linenomath}
This completes the proof.
\end{proof}

\begin{proof}[Proof of Proposition \ref{prop:algext}]

We recover $\bm{M}^\psi$ by inserting the parameters from \eqref{eq:sol} into \eqref{eq:interior_pars}. To infer the invariants we first look at the equality conditions. We do this representatively by looking at $\overline{\bm{M}}^\alpha=\widetilde{\bm{M}}^\alpha$. In particular we look at
\begin{linenomath}
\begin{align*}
\oM^\alpha_{11}-\oM^\alpha_{01}&=\widetilde{M}^\alpha_{11}-\widetilde{M}^\alpha_{01},\quad
\oM^\alpha_{11}+\oM^\alpha_{01}=\widetilde{M}^\alpha_{11}+\widetilde{M}^\alpha_{01},\\
\end{align*}
\end{linenomath}
and thus
\begin{linenomath}
\begin{align*}
\frac{\sqrt{\chi_{\alpha\beta\gamma}}}{\coa_{\beta\gamma}}&=\frac{\sqrt{\chi_{\alpha\beta\delta}}}{\coa_{\beta\delta}},\quad
\frac{\coa_{\alpha\beta\gamma}}{\coa_{\beta\gamma}}=\frac{\coa_{\alpha\gamma\delta}}{\coa_{\beta\delta}},\\
\frac{\coa_{\alpha\beta\gamma}^2+4\coa_{\alpha\beta}\coa_{\alpha\gamma}\coa_{\beta\gamma}}{\coa_{\beta\gamma}^2}&=\frac{\coa_{\alpha\beta\delta}^2+4\coa_{\alpha\beta}\coa_{\alpha\delta}\coa_{\beta\delta}}{\coa_{\beta\delta}^2},\quad
\frac{\coa_{\alpha\beta\gamma}}{\coa_{\beta\gamma}}=\frac{\coa_{\alpha\gamma\delta}}{\coa_{\beta\delta}},\\
\frac{\coa_{\alpha\gamma}}{\coa_{\beta\gamma}}&=\frac{\coa_{\alpha\delta}}{\coa_{\beta\delta}},\quad
\frac{\coa_{\alpha\beta\gamma}}{\coa_{\beta\gamma}}=\frac{\coa_{\alpha\gamma\delta}}{\coa_{\beta\delta}},\\
\frac{\coa_{\alpha\beta\gamma}}{\coa_{\alpha\beta\delta}}&=\frac{\coa_{\beta\gamma}}{\coa_{\beta\delta}}=\frac{\coa_{\alpha\gamma}}{\coa_{\alpha\delta}}.
\end{align*}
\end{linenomath}
Looking at $\overline{\bm{M}}^\beta=\widetilde{\bm{M}}^\beta$ yields the same equalities. Reproducing the calculations for $\widehat{\bm{M}}^\gamma=\check{\bm{M}}^\gamma$ yields the invariants $f_1$ to $f_3$.

For the inference of $f_0$ observe that the equation system \eqref{eq:leaftotal} can be written in a marginalised form, i.e. one replaces the equations in $(p_{abcd})_{a,b,c,d\in\{0,1\}}$ by the linear transforms $\era_{\alpha\beta\gamma\delta},\,\era_{\alpha\beta\gamma},\,\era_{\alpha\beta\delta},\,\era_{\alpha\gamma\delta},\,\era_{\beta\gamma\delta},\era_{\alpha\beta},\,\era_{\alpha\gamma},\,\era_{\alpha\delta},\,\era_{\beta\gamma},\era_{\beta\delta},\era_{\gamma\delta},\,\era_\alpha,\,\era_\beta,\,\era_{\gamma}$ and $\era_{\delta}$.

We immediately see that all terms but $\era_{\alpha\beta\gamma\delta}$ are covered by our investigation of the tripod case. We insert the parameters obtained in \eqref{eq:sol} and \eqref{eq:interior} into the equation for $\era_{\alpha\beta\gamma\delta}$ to get:
\begin{linenomath}
\begin{align*}
\era_{\alpha\beta\gamma\delta}&=(1-q^\wurzel_1)\oM^\alpha_{01}\oM^\beta_{01}((1-M^\psi_{01})\wM^\gamma_{01}\wM^\delta_{01}+M^\psi_{01}\wM^\gamma_{11}\wM^\delta_{11})\\
&+q^\wurzel_1\oM^\alpha_{11}\oM^\beta_{11}((1-M^\psi_{11})\wM^\gamma_{01}\wM^\delta_{01}+M^\psi_{11}\wM^\gamma_{11}\wM^\delta_{11}).
\end{align*}
\end{linenomath}
Reordering and restructuring this equation eventually yields invariant $f_0$. This completes the proof.
\end{proof}

\begin{proof}[Proof of Proposition \ref{prop:stoch}]
Theorem \ref{thm:stoch} covers the first part of the Proposition. The remaining inequalities are obtained by bounding \eqref{eq:interior} between 0 and 1 and use the fact that the covariances are always positive with Lemma \ref{lem:indi}(1):
\begin{linenomath}
\begin{align*}
-1&\le\frac{\coa_{\alpha\delta}\coa_{\alpha\beta\gamma}-\coa_{\alpha\beta}\coa_{\alpha\gamma\delta}-\coa_{\alpha\delta}\sqrt{\chi_{\alpha\beta\gamma}}}{\coa_{\alpha\beta}\sqrt{\chi_{\alpha\gamma\delta}}}\le1,\\
\coa_{\alpha\beta}(\coa_{\alpha\gamma\delta}-\sqrt{\chi_{\alpha\gamma\delta}})&\le
\coa_{\alpha\delta}(\coa_{\alpha\beta\gamma}-\sqrt{\chi_{\alpha\beta\gamma}})\le\coa_{\alpha\beta}(\coa_{\alpha\gamma\delta}+\sqrt{\chi_{\alpha\gamma\delta}}),\\
-1&\le\frac{\coa_{\alpha\delta}\coa_{\alpha\beta\gamma}-\coa_{\alpha\beta}\coa_{\alpha\gamma\delta}+\coa_{\alpha\delta}\sqrt{\chi_{\alpha\beta\gamma}}}{\coa_{\alpha\beta}\sqrt{\chi_{\alpha\gamma\delta}}}\le1,\\
\coa_{\alpha\beta}(\coa_{\alpha\gamma\delta}-\sqrt{\chi_{\alpha\gamma\delta}})&\le
\coa_{\alpha\delta}(\coa_{\alpha\beta\gamma}+\sqrt{\chi_{\alpha\beta\gamma}})\le\coa_{\alpha\beta}(\coa_{\alpha\gamma\delta}+\sqrt{\chi_{\alpha\gamma\delta}}),\\
\end{align*}
\end{linenomath}
\end{proof}


\begin{thebibliography}{26}
\providecommand{\natexlab}[1]{#1}
\providecommand{\url}[1]{\texttt{#1}}
\expandafter\ifx\csname urlstyle\endcsname\relax
  \providecommand{\doi}[1]{doi: #1}\else
  \providecommand{\doi}{doi: \begingroup \urlstyle{rm}\Url}\fi

\bibitem[Bryant et~al.(2005)Bryant, Galtier, and Poursat]{bryant2005b}
David Bryant, Nicolas Galtier, and Marie-Anne Poursat.
\newblock Likelihood calculation in molecular phylogenetics.
\newblock In Olivier Gascuel, editor, \emph{Mathematics of Evolution and
  Phylogeny}, chapter~2, pages 33--62. Oxford University Press, 2005.

\bibitem[Chang(1996)]{chang1996}
Joseph~T. Chang.
\newblock Full reconstruction of {Markov} models on {{Evolutionary} Trees}:
  Identifiability and consistency.
\newblock \emph{Math. Biosci.}, 137:\penalty0 51--73, 1996.
\newblock URL \url{http://dx.doi.org/10.1016/S0025-5564(96)00075-2}.

\bibitem[Lazarsfeld and Henry(1968)]{lazarsfeld1968}
Paul~F. Lazarsfeld and Neil~W. Henry.
\newblock \emph{{Latent Structure Analysis}}.
\newblock Houghton, Mifflin, New York, 1968.

\bibitem[Pearl and Tarsi(1986)]{pearl1986}
Judea Pearl and Michael Tarsi.
\newblock Structuring causal trees.
\newblock \emph{J. Complexity}, 2:\penalty0 60--77, 1986.
\newblock URL \url{http://dx.doi.org/10.1016/0885-064X(86)90023-3}.

\bibitem[Allman et~al.(2009)Allman, Matias, and Rhodes]{allman2009}
Elizabeth~S. Allman, Catherine Matias, and John~A. Rhodes.
\newblock Identifiability of parameters in latent structure models with many
  observed variables.
\newblock \emph{Ann. Stat.}, 37\penalty0 (6A):\penalty0
  3099--3132, 2009.
\newblock URL \url{http://dx.doi.org/10.1214/09-AOS689}.

\bibitem[Allman and Rhodes(2008)]{allman2008}
Elizabeth~S. Allman and John~A. Rhodes.
\newblock Phylogenetic ideals and varieties for the general Markov model.
\newblock \emph{Adv. Appl. Math.}, 40\penalty0 (2):\penalty0
  127--148, 2008.
\newblock URL \url{http://dx.doi.org/10.1016/j.aam.2006.10.002}.

\bibitem[Allman and Rhodes(2003)]{allman2003}
Elizabeth~S. Allman and John~A. Rhodes.
\newblock Phylogenetic invariants for the general {Markov} model of sequence
  mutation.
\newblock \emph{Math. Biosci.}, 186\penalty0 (2):\penalty0 113--144,
  December 2003.
\newblock URL \url{http://dx.doi.org/10.1016/j.mbs.2003.08.004}.

\bibitem[Baake(1998)]{baake1998a}
Ellen Baake.
\newblock What can and what cannot be inferred from pairwise sequence
  comparisons?
\newblock \emph{Math. Biosci.}, 154\penalty0 (1):\penalty0 1--21,
  1998.
\newblock URL \url{http://dx.doi.org/10.1016/S0025-5564(98)10044-5}.
\bibitem[Sturmfels and Sullivant(2005)]{sturmfels2005}
Bernd Sturmfels and Seth Sullivant.
\newblock Toric ideals of phylogenetic invariants.
\newblock \emph{J. Comput. Biol.}, 12\penalty0 (4):\penalty0
  457--481, 2005.
\newblock URL \url{http://dx.doi.org/10.1089/cmb.2005.12.457}.

\bibitem[Lake(1987)]{lake1987}
James~A. Lake.
\newblock A rate-independent technique for analysis of nucleic acid sequences:
  evolutionary parsimony.
\newblock \emph{Mol. Biol. Evol.}, 4\penalty0 (2):\penalty0
  167--191, 1987.
\newblock URL \url{http://mbe.oxfordjournals.org/content/4/2/167.abstract}.

\bibitem[Cavender and Felsenstein(1987)]{cavender1987}
James~A. Cavender and Joseph Felsenstein.
\newblock Invariants of phylogenies in a simple case with discrete states.
\newblock \emph{J. Classif.}, 4:\penalty0 57--71, 1987.
\newblock URL \url{http://dx.doi.org/10.1007/BF01890075}.

\bibitem[Evans and Speed(1993)]{evans1993}
Steven N. Evans and Terrence P. Speed.
\newblock Invariants of some probability models used in phylogenetic inference.
\newblock \emph{Ann. Stat.}, 21\penalty0(1):\penalty0 355-377, 1993.
\newblock URL \url{http://dx.doi.org/10.1214/aos/1176349030}.

\bibitem[Sumner et~al.(2008)Sumner, Charleston, Jermiin, and
  Jarvis]{sumner2008}
Jeremy~G. Sumner, Michael~A. Charleston, Lars~S. Jermiin, and Peter~D. Jarvis.
\newblock Markov invariants, plethysms, and phylogenetics.
\newblock \emph{J. Theor. Biol.}, 253:\penalty0 601--615, 2008.
\newblock URL \url{http://dx.doi.org/10.1016/j.jtbi.2008.04.001}.

\bibitem[Sumner and Jarvis(2009)]{sumner2009}
Jeremy~G. Sumner, and Peter~D. Jarvis.
\newblock Markov invariants and the isotropy subgroup of a quartet tree.
\newblock \emph{J. Theor. Biol.}, 258:\penalty0 302--310, 2009.
\newblock URL \url{http://dx.doi.org/10.1016/j.jtbi.2009.01.021}.

\bibitem[Zwiernik and Smith(2010)]{zwiernik2010}
Piotr Zwiernik and Jim~Q. Smith.
\newblock Tree-cumulants and the identifiability of bayesian tree model.
\newblock arXiv:1004.4360v1, 2010.
\newblock URL \url{http://arxiv.org/abs/1004.4360}.

\bibitem[Matsen(2009)]{matsen2009}
Frederick~A. Matsen.
\newblock Fourier transform inequalities for phylogenetic trees.
\newblock \emph{IEEE ACM T. Comput. Bi.}, 6\penalty0 (1):\penalty0 89--95, 2009.
\newblock URL
  \url{http://dx.doi.org/10.1109/TCBB.2008.68}.

\bibitem[Yang(2000)]{yang2000}
Ziheng Yang.
\newblock Complexity of the simplest phylogenetic estimation problem.
\newblock \emph{Proc. Royal Soc. B}, 267:\penalty0 109--116, 2000.
\newblock URL \url{http://dx.doi.org/10.1098/rspb.2000.0974}

\bibitem[Cavender(1978)]{cavender1978}
James~A. Cavender.
\newblock Taxonomy with confidence.
\newblock \emph{Math. Biosci.}, 40:\penalty0 271--280, 1978.
\newblock URL \url{http://dx.doi.org/10.1016/0025-5564(78)90089-5}.

\bibitem[Farris(1973)]{farris1973}
James~S. Farris.
\newblock A probability model for inferring evolutionary trees.
\newblock \emph{Syst. Zool.}, 22\penalty0 (3):\penalty0 250--256, 1973.
\newblock URL \url{http://www.jstor.org/stable/2412305}.

\bibitem[Neyman(1971)]{neyman1971}
Jerzy Neyman.
\newblock Molecular studies of evolution: A source of novel statistical
  problems.
\newblock In S.~dasGupta and J.~Yackel, editors, \emph{Statistical Decision
  Theory and Related Topics}, pages 1--27. Academic Press, New York, 1971.

\bibitem[Hendy and Penny(1989)]{hendy1989}
Michael~D. Hendy and David Penny.
\newblock A framework for the quantitative study of evolutionary trees.
\newblock \emph{Syst. Zool.}, 38\penalty0 (4):\penalty0 297--309, 1989.
\newblock URL \url{http://dx.doi.org/10.2307/2992396}.

\bibitem[Szekely et~al.(1993)Szekely, Steel, and Erd\"os]{szekely1993}
L\`aszl\'o Szekely, Mike~A. Steel, and Peter~L. Erd\"os.
\newblock Fourier calculus on evolutionary trees.
\newblock \emph{Adv. Appl. Math.}, 14\penalty0 (2):\penalty0
  200--210, 1993.
\newblock URL \url{http://dx.doi.org/10.1006/aama.1993.1011}

\bibitem[Lauritzen(1996)]{lauritzen1996}
Steffen~L. Lauritzen.
\newblock \emph{Graphical Models}.
\newblock Oxford Stastical Science Series. Clarendon Press, Oxford, 1996.
\newblock ISBN 0-19-852219-3.
\newblock URL \url{http://ukcatalogue.oup.com/product/9780198522195.do}.

\bibitem[Rayner and Beh(2009)]{rayner2009}
J.~C.~W. Rayner and Eric~J. Beh.
\newblock Towards a better understanding of correlation.
\newblock \emph{Stat. Neerl.}, 63\penalty0 (3):\penalty0 324--333,
  2009.
\newblock URL
  \url{http://dx.doi.org/10.1111/j.1467-9574.2009.00425.x}.
  
\bibitem[Steel(1994)]{steel1994}
Mike A. Steel.
\newblock Recovering a tree from the leaf colourations it generates under a Markov model.
\newblock \emph{Appl. Math. Letter}, 7\penalty0 (2), 19--23, 1994.
\newblock URL \url{http://dx.doi.org/10.1016/0893-9659(94)90024-8}.

\bibitem[Semple and Steel(2003)]{semple2003}
Charles Semple and Mike Steel.
\newblock \emph{Phylogenetics}.
\newblock Oxford Lectures Series in Mathematics and its Applications. Oxford
  University Press, 2003.
\newblock ISBN 0-19-850942-1.

\bibitem[Chor et~al.(2000)Chor, Hendy, Holland, and Penny]{chor2000}
Benny Chor, Michael~D Hendy, Barbara~R Holland, and David Penny.
\newblock Multiple maxima of likelihood in phylogenetic trees: An analytic
  approach.
\newblock \emph{Mol. Biol. Evol.}, 17\penalty0 (10):\penalty0
  1529--1541, 2000.
\newblock URL \url{http://mbe.oxfordjournals.org/content/17/10/1529.full}.

\change{\bibitem[Churchill and von Haeseler (1993)]{churchill1993}
Gary A. Churchill and Arndt von Haeseler.
\newblock Network analysis for sequence evolution.
\newblock \emph{J. Mol. Evol.}, 37:\penalty0 77--85, 1993.
\newblock URL \url{http://dx.doi.org/1007/BF00170465}}

\change{\bibitem[Strimmer and Moulton(2000)Strimmer and Moulton]{strimmer2000}
Korbinian Strimmer and Vincent Moulton.
\newblock Likelihood Analysis of Phylogenetic Networks Using Directed Graphical Models.
\newblock \emph{Mol. Biol. Evol.}, 17\penalty0 (6):\penalty0
 875--881, 2000.
\newblock URL \url{http://mbe.oxfordjournals.org/content/17/6/875.short}.}

\bibitem[Bryant(2005)]{bryant2005}
David Bryant.
\newblock Extending tree models to splits networks.
\newblock In \emph{Algebraic Statistics for Computational Biology}, chapter~17,
  pages 320--332. Cambridge University Press, 2005.

\end{thebibliography}
\end{document}